\documentclass{article}
\usepackage{dblfloatfix}
\usepackage{listings}
\usepackage{algorithm}
\usepackage{algpseudocode}
\usepackage[english]{babel}

\usepackage[letterpaper,top=2cm,bottom=2cm,left=3cm,right=3cm,marginparwidth=1.75cm]{geometry}

\usepackage[utf8]{inputenc}
\usepackage{mathrsfs}
\usepackage{dsfont}
\usepackage{amsfonts}
\usepackage{complexity}
\usepackage{latexsym}
\usepackage{amssymb}    
\usepackage{amsmath}
\usepackage{graphicx}
\usepackage[colorlinks=true, allcolors=blue]{hyperref}

	\newcommand{\bbN}{\mathbb{N}}

\newcommand{\mcA}{\mathcal{A}}	\newcommand{\mcB}{\mathcal{B}}
\newcommand{\mcC}{\mathcal{C}}	
	\newcommand{\mcF}{\mathcal{F}}

	\newcommand{\mcL}{\mathcal{L}}
	
\newcommand{\mcO}{\mathcal{O}}

\newcommand{\bmW}{W}

\newcommand{\eps}{\epsilon}

\newcommand{\ceil}[1]{\lceil #1 \rceil}
\newcommand{\floor}[1]{\lfloor #1 \rfloor}

\newcommand{\cltonp}{\CL^{\NP}_{2\text{-round}}}

\renewcommand{\P}{\mathsf{P}}

\renewcommand{\epsilon}{\varepsilon}

\renewcommand{\epsilon}{\varepsilon}
\newcommand{\ignore}[1]{}
\newcommand{\Wt}{{W}}
\newcommand{\Wtg}[2]{\Wt|_{#1\gets #2}}

\usepackage{cleveref}
\usepackage{babel}
\usepackage{amsthm}
\usepackage{thmtools,thm-restate}
\usepackage{lineno}
\newtheorem{definition}{Definition}[section]
\newtheorem{theorem}{Theorem}[section]
\newtheorem{claim}[theorem]{Claim}
\newtheorem{corollary}{Corollary}[theorem]
\newtheorem{lemma}[theorem]{Lemma}
\newtheorem{proposition}[theorem]{Proposition}

\newtheorem{remark}[theorem]{Remark}
\newtheorem{result}[theorem]{Result}

\title{Derandomizing Isolation In Catalytic Logspace}
\author{V. Arvind, Srijan Chakraborty, Samir Datta}

\begin{document}
\maketitle

\abstract 
A language is said to be in catalytic logspace if we can test membership
using a deterministic logspace machine that has an additional read/write tape
filled with arbitrary data whose contents have to be restored to their 
original value at the end of the computation. The model of \emph{catalytic
computation}  was introduced by Buhrman et al [STOC 2014].

As our first result, we obtain a catalytic logspace algorithm for computing a 
minimum weight witness to a search problem, with small weights, provided the 
algorithm is given oracle access for the corresponding weighted decision problem.
This reduction of Search to Decision is similar in spirit to the Anari-Vazirani result 
[ITCS 2020], adapted to catalytic logspace (instead of $\NC$), that applies to 
\emph{a large class} of search problems (instead of just perfect matching).
The search to decision reduction crucially uses the isolation lemma of 
Mulmuley et al [Combinatorica 1987], but sidesteps the Anari-Vazirani machinery 
[J. ACM 2020, ITCS 2020] for showing planar perfect matching in $\NC$ and 
the search to decision reduction [ITCS 2020] for perfect matching.
In particular, our reduction yields catalytic logspace algorithms for the search
versions of the following three problems: planar perfect matching, planar exact 
(or red-blue) perfect matching and weighted arborescences in general weighted 
digraphs. 


Our second set of results concern the oracle complexity class
$\cltonp$ (defined in Section~\ref{sec:oracle}). We show that $\cltonp$ 
contains $\fF$ and the complexity classes 
$\BPP,\MA$ and $\ZPP^{\NP[1]}$. While $\fF$ is shown to be in $\cltonp$ using 
the isolation lemma, the other three containments, while based on the compress-or-random 
technique, use the Nisan-Wigderson [JCSS 1994] (and Impagliazzo-Wigderson [STOC 1997])
hardness vs randomness based pseudo-random generator. These containments show that 
$\cltonp$ resembles $\ZPP^\NP$ more than $\P^\NP$, providing some weak evidence that 
$\CL$ is more like $\ZPP$ than $\P$.

For our third set of results we turn to isolation well inside catalytic classes. 
We consider the catalytic class $\C\TISP[\poly(n),\log n,\log^2 n]^{\UL}$ which uses only 
$\mcO(\log^2{n})$ catalytic space along with $\mcO(\log{n})$ work space, with $\UL$ given as an oracle, while staying
inside polynomial time. We show that this class contains reachability, and therefore $\NL$. This is a catalytic version of the result of van Melkebeek \&
Prakriya [SIAM J. Comput. 2019]. Building on their result, we also show a tradeoff
between the workspace of the unambiguous oracle machine and the catalytic space i.e. placing reachability in
$\C\TISP[\poly(n),\log{n},\log^{2-\alpha}{n}]^{\U\TISP[\poly(n),\log^{1+\alpha}n]}$ where 
$0 \leq \alpha \leq 0.5$. Finally, we extend these catalytic upper bounds to $\LogCFL$
by placing it in $\C\TISP[\poly(n),\log{n},\log^{2-\alpha}{n}]^{\U\AuxPDA\text{-}\TISP[\poly(n),\log^{1+\alpha}n]}$ 
for $\alpha \in [0,0.5]$.


\section{Introduction}
Space-bounded computation is conventionally studied using the Turing
machine model that consists of a read-only input tape and a
space-bounded work-tape, sometimes with additional resources such as
randomness or non-determinism. Catalytic computation is a relatively
new paradigm of space-bounded computation introduced about a decade
ago by Buhrman, Cleve, Koucky, Loff, Speelman \cite{BCKLS14} that
allows the Turing machine additional work-space: this is a \emph{full
tape} of storage (often exponentially larger than the work tape). The
machine can use this additional work-tape for both reads and
writes. But the important restriction is that the initial tape
contents must be restored when the machine halts. In a sense, this
additional tape catalyses the computation without any change to it
after the machine halts. Hence, this extra storage tape is called the
catalytic tape. We should mention here that this model of computation
is inspired by a seminal paper by Ben-Or and Cleve \cite{BenOr-Cleve}
in algebraic computation showing that $3$-register algebraic
straight-line programs suffice for evaluating arithmetic formulas.

The first paper \cite{BCKLS14} proved that catalytic logspace $\CL$ is
apparently more powerful than conventional logspace $\Log$. Indeed, it
even contains non-deterministic Logspace $\NL$. In fact, $\CL$ can
simulate $\TC^1$ i.e. languages recognized by threshold circuits of
logarithmic depth. On the other hand, we do not even know whether
$\CL$ is contained in polynomial time with the best known upper bound
for $\CL$ being $\ZPP$.

Catalytic logspace, $\CL$ can be augmented with nondeterminism
(essentially, an $\NL$ machine with a polynomially bounded catalytic
tape) to define the class $\CNL$.  Randomized catalytic logspace can
also be considered (essentially, languages accepted by $\BPL$ machines
with a polynomially bounded catalytic tape) to define $\CBPL$. These
classes were studied in \cite{DGJST,CLMP}. Unambiguous $\CL$
(languages accepted by unambiguous logspace machines with a
polynomially bounded catalytic tape) is considered in \cite{GJST,DGJST25}. It
turns out, as shown by Koucky, Mertz, Pyne and Sami in \cite{KMPS},
that these new classes coincide with $\CL$. These properties show
$\CL$ to be a robust complexity class.


A different thread of research in catalytic computation, that was
focused on improved catalytic space upper bounds for the $\TreeEval$,
led to a major breakthrough by Cook and Mertz \cite{CM21,CM23}. They
showed that $\TreeEval$ requires only catalytic space
$\mcO(\log{n}\log{\log{n}})$ (and $\mcO(\log n\log\log n)$ work space) in \cite{CM23}. This is an important
example of a catalytic bound that yields the best known space bound
(catalysed or otherwise) on a problem. Using this result, \cite{Will} proved that every multitape Turing machine running in time $t$ can
be simulated in space only $\mcO(\sqrt {t\log t})$.


Building on Nisan's seminal result that $\BPL \subseteq \SC^2$, Pyne
\cite{Pyne} has shown that $\BPL$ is contained in simultaneous $\log$
space, catalytic $\log^2$ space and polynomial time i.e.
$\CTISP[n^{\mcO(1)},\log{n},\log^2{n}]$. Further, he has shown a
smooth tradeoff between space and catalytic space i.e. it was shown
that $\BPL$ is contained in
$\CSPACE[\log^{1+\alpha}{n},\log^{2-\alpha}{n}]$ for every $\alpha \in
[0,0.5]$.

As a new addition to upper bound results, Agarwala and Mertz \cite{AM}
have shown that Bipartite Matching Search is in $\CL$. This is the
first problem not known to be in deterministic $\NC$ that is placed
inside $\CL$ making it a remarkable result.

The last two upper bound results, mentioned above, use the so-called
\emph{compress-or-random} framework which interprets the catalytic
tape contents as random bits -- if this assumption is correct and the
bits ``behave like'' random bits as required by the algorithm, then the
randomized algorithm succeeds and the problem is solved. On the other
hand, if the bits are not really random in a prescribed sense, then
the catalytic tape can be efficiently compressed in a manner that they
can be efficiently decompressed at the end of the algorithm. This
compression frees up space on the catalytic tape that can be
essentially used for any polynomial-space computation.


The field of Catalytic Logspace has seen numerous other contributions
in the last year or so. Without being comprehensive, here is a sampling
\cite{AAV,AFMSV,CGMPS,TP,Will}. We would like to point out two papers that are
related to our work. The first one, \cite{AAV} shows a generalisation of
\cite{AM} by proving that linear matroid intersection is in $\CL$. The other, \cite{CGMPS} initiates the study of oracle based catalytic classes by proving the
existence of an oracle $O$ such that $\CL^O=\EXP^O$ thus showing that a
relativizing proof of $\CL = \P$ cannot exist.


\subsection{Technical Overview and Consequences}\label{sub:tech}

\paragraph*{Reducing search to weighted decision} 

Mulmuley, Vazirani and Vazirani \cite{MVV} gave a randomized $\NC$ (also $\ZPP$) algorithm 
that reduces search to weighted decision with polynomially bounded weights in the context 
of Perfect Matching. The reduction works by using the isolation lemma to isolate a minimum 
weight matching that can then be extracted with multiple parallel queries to the decision 
oracle.  

This general recipe works for broader class of search problems. Indeed, the recipe can be 
applied to reduce search to weighted decision with polynomially bounded 
weights for every language with a witness verifiable in $\P$ (i.e.  for languages in $\NP$). 
In fact, in \Cref{sec:three}, we show in a more general setting that search reduces to 
weighted decision with polynomially bounded weights, where the reduction is computable in
\emph{catalytic logspace}. Thus, in a sense, we can derandomize the isolation lemma in 
these applications using catalytic logspace.




This presents some (weak) evidence that $\CL$ is closer to $\ZPP$ than $\P$.

This also enables us to place several problems in $\CL$. Important 
examples include  search for the following:
\begin{itemize}
\item Planar Perfect Matching: known to be in $\NC$ by recent breakthroughs \cite{AV18,Sank} which we do not need to invoke in our proof.  \item Planar Exact Matching: in $\P$ \cite{Yuster} but not known to be in $\NC$
 or even $\qNC$, 
\item Min-wt Arborescence in digraphs: known to be in $\NC$ \cite{Lovasz} 
\item Some other examples follow from \cite{Oki}, from various applications of Linear Pfaffian Matroid Parity.
\end{itemize}

Let $\mcL$ be a language in $\NP$. Thus, there exists a predicate $R_\mcL$ such that
a word $x \in \mcL$ iff there exists a witness $y$ satisfying that $R_\mcL(x,y)$ holds 
true for
$x \in \{0,1\}^*$  and $y \in \{0,1\}^{m(|x|)}$ where $m(\cdot)$ is a polynomial.
The decision version for $\mcL$, on input $x$ asks if there exists a witness $y$ 
such that $(x,y)$ satisfies the relation 
$R_\mcL$. The search version requires to find such a $y$. One strategy to find a
witness is to isolate one by using a weight function such that the min weight
witness becomes unique and then can be extracted using a weighted decision oracle. 
We amplify on this strategy below.

Given a $y\in \{0,1\}^{m}$, consider the set $F_y=\{e:y_e=1\}\subseteq [m]$. 
A weight function $\Wt:[m]\to \{0,1,\ldots,{\poly(n)}\}$,  naturally extends to
subsets $S \subseteq [m]$ as $\Wt(S) = \sum_{e\in S}{\Wt(e)}$.
The weighted decision oracle  takes as input $(x,\Wt,w)$ and outputs a $1$ iff
there is a witness $y$ of weight $\Wt(F_y) \leq w$.
We can find $\min_{y:R_\mcL(x,y)=1}W(F_y)$ by sequentially querying the decision 
oracle for all $w \in [0,\poly(n)]$ until the oracle returns $1$. 

Under the assumption that the minimum weight $w$ of $y$ is unique we can use the oracle
to even find this $y$ by modifying the weight assignment $\Wt$. In particular
we define the weight function $\Wtg{e}{a}$ as being identical to $\Wt$ except
that for bit $e$ the weight is set to $a$. Since we assume that $y$ is unique so
$e \in F_y$ iff $\min_{y:R_\mcL(x,y)=1}\Wtg{e}{\infty}(F_y) > w$. Notice that this
is similar to the strategy used by \cite{MVV} to isolate and extract a perfect matching,
The difference is that instead of deleting $e$ we set its weight to $\infty$.

Thus, we are left to deal with the case that the isolating assumption on the minimum
weight $y$ does not hold. In this case, we can detect a ``threshold'' element $e$
that occurs in one minimum weight witness but does not occur in another minimum weight 
witness. A threshold element indicates that the weight 
function $\Wt$ is compressible because we can recreate the weight $\Wt(e)$  knowing
only the weights $\Wt(e')$ for $e' \neq e$ and in fact from $\Wtg{e}{0}$
and $\Wtg{e}{\infty}$ simply as the the difference of the weights of the
minimum weight witnesses under the preceding two weight assignments.

In the catalytic logspace setting we consider weight assignments, one by one, 
from the catalytic storage. If we find an isolating assignment we can 
use the oracle to find the minimum weight witness. On the other hand if we find
a non-isolating weight assignment we can compress the weight function by $\mcO(\log{n})$ bits
with access to the weighted decision oracle. Thus, examining polynomially many
weight assignments, if  we are unable to find an isolating assignment among them, 
we would have saved polynomial space by the compression sketched above, and can use 
any polynomial space algorithm to find the solution. As $\mcL \in \NP$ such an algorithm 
clearly exists.

To summarise, the main result here is the following:
\begin{result} (\Cref{thm:srch})
For polynomially bounded weights, min-weight-search reduces to 
  weighted-decision for languages in $\NP$,
  under Turing reductions computable in $\CL$.
\end{result}

\paragraph*{Catalytic Logspace with an $\NP$ oracle}

It is shown in \cite{BCKLS14} that $\CL\subseteq\ZPP$. In fact, this containment relativizes.
That is, for any oracle $A$ we have $\CL^A\subseteq \ZPP^A$. In particular, we have 
$\CL^{\NP}\subseteq \ZPP^{\NP}$. In this section, we consider a variant of the class $\CL^{\NP}$ 
in which the queries are asked in a non-adaptive fashion in rounds, which we need to carefully define.

We first observe that the reduction from search to weighted decision directly implies that $\fF\in\CL^{\NP}$ where $\fF$ is the search version of $\SAT$. That is, given a CNF formula $\phi$, $\fF$ is the problem of computing a satisfying assignment for a given input formula $\phi$, if it is satisfiable. This easily follows as the weighted decision queries can be simulated using a $\NP$ oracle. It is shown in \cite{VV} that there is a polynomial-time randomized reduction from $\SAT$ to unique-$\SAT$, where unique-$\SAT$ asks whether a formula $\phi$ is in $\SAT$ or not, conditioned on $\phi$ having at most one satisfying assignment. The isolating lemma of \cite{MVV} can be used to provide a different proof that $\SAT$ is randomly reducible to unique-$\SAT$. Our result on $\fF$ is essentially a derandomization of this randomized reduction based on catalytic isolation. With a careful analysis of the above reduction, we show that the search problem $\fF$ can be solved using a $\CL$ machine with 
\emph{two parallel rounds} of polynomially many $\NP$ queries.\footnote{The definition of the $\CL$ oracle machine models are presented in Section~\ref{sec:oracle}} We name this $\CL$-oracle class $\cltonp$.

Next, we show that any randomized polynomial-time algorithm can be derandomized in $\cltonp$, i.e. $\BPP\subseteq \cltonp$. The proof of this containment gives a different way of using catalytic 
space. We now sketch the proof: This proof is based on the hardness versus randomness trade-off
introduced in Nisan and Wigderson (strengthened subsequently by Impagliazzo-Wigdereson) 
\cite{NW,IW}. They designed a pseudorandom generator (PRG) that
can stretch $\mcO(\log n)$ size random seeds to $n^{\mcO(1)}$ pseudorandom strings that can fool all
$\BPP$ algorithms, under a hardness assumption for a language in $\E$ (which denotes DTIME$[2^{\mcO(n)}]$). More precisely, the hardness assumption \cite{IW} is that there is
a language $L\in \DTIME[2^{\mcO(n)}]=\E$ such that for all but finitely many $n$, any boolean
circuit of size $2^{\eps n}$, for some $\eps<1$, will fail to accept $L^{=n}$, that is $L$ 
for length $n$ inputs. 

This idea was used by Goldreich-Zuckerman \cite{GZ} to show that $\BPP$ is contained in $\ZPP^\NP$, where the oracle $\ZPP$ machine first randomly guesses the truth-table of a hard boolean function $f$ and then the $\NP$ oracle is used to verify its hardness (with a single $\NP$ query). Having found a hard function $f$, the Nisan-Wigderson pseudorandom generator can now be constructed with $f$ and then used to derandomize any $\BPP$ computation. We adapt this approach to showing $\BPP$ is contained in $\cltonp$. We assume that the catalytic tape is given to us as a collection of truth tables from $\{0,1\}^{\log n}$ to $\{0,1\}$, each represented as an $n$ bit string. We query the $\NP$ oracle, to check if there is a size $n^{\eps}$ circuit evaluating one of these truth tables. If so, we replace the truth table, say, $T$ with the corresponding circuit, say, $C$. Observe that $C$ requires $n^{\eps}$ bits to store, so we have in the process saved some space. If by the end of our computations, at least one of the sampled truth tables is hard for circuits of size $n^{\eps }$, we have found a pseudorandom generator, and can derandomize $\BPP$ in $\P$. Thus, with one more $\NP$ query, we are done. If none of the truth-tables is hard for $n^{\eps}$ size circuits, then we have saved enough space to derandomize $\BPP$ with a brute-force search in $\PSPACE$. Finally, we can replace the circuits with their corresponding truth-tables. There is an issue with the above arguement: When we ask the $\NP$ oracle whether or not a truth table is computable using circuits of $n^{\eps}$ size, how do we get access to one such circuit explicitly? To solve this problem, we can use the reduction showing $\fF\in\cltonp$ to find the circuit corresponding to the truth-table on a catalytic tape with one parallel round of suitable $\NP$ queries. Essentially, it turns out that all required $\NP$ queries can be made into two parallel rounds, thus showing that $\BPP\subseteq \cltonp$. We show similar inclusions for the classes $\MA$ and $\ZPP^{\NP[1]}$ with minor modifications to the above argument.


To summarise:
\begin{result} (\Cref{thm:fsat,thm:BppMaZppToNp})
$\fF,\BPP,\MA,\ZPP^{\NP[1]}$ are in $\cltonp$.
\end{result}

\paragraph*{Catalytic isolation for Reachability and Semi-unbounded circuit evaluation}
\medspace

In the previous sections we saw that the isolation lemma of \cite{MVV} can be
derandomized in the context of $\CL$ and $\cltonp$. In this section we 
derandomize the isolation lemma as applied to reachability and evaluation of
semi-unbounded circuits. Notice that both these problems are well inside 
the entire $\CL$ (since $\NL \subseteq \LogCFL \subseteq \TC^1 \subseteq \CL$)
 so we have to explore subclasses of $\CL$ for this to make sense.
For this we consider variants of $\CL$
by further restricting the catalytic space and with access to unambigious low space oracles. The ensuing fine grained classes form upper bounds for several 
natural and 
interesting problems/classes. This includes the result from \cite{Pyne} 
simulating $\BPL$ in catalytic space $\mcO(\log^2{n})$ while persisting with 
$\mcO(\log{n})$ work space and polynomial time. 
Also, an upper bound for $\BPL$ that is a tradeoff result between 
catalytic and work space without preserving polynomial time. A third example is
$\TreeEval$  for which, an intriguing  $\mcO(\log{n}\log{\log{n}})$
catalytic and work space bound is shown in \cite{CM23}.

However, notice that it is not viable to show $\NL$ in small catalytic space
and work space while preserving polynomial time -- such a proof would imply
$\NL \subseteq \SC$ resolving a venerable open problem. 
In this backdrop, we aim to show similar results for $\NL$ and $\LogCFL$ by 
modifying the techniques originating from \cite{AR} which was in turn inspired
by \cite{GalW}.
It was shown in \cite{AR} that logspace non-determinism can be made
unambiguous at least non-uniformly by crucially using the isolation lemma. In the same spirit, we show that $\NL$ is contained in the catalytic class with $\mcO(\log^2{n})$ catalytic 
space and $\mcO(\log{n})$ work space with access to a $\UL$ oracle (notice that this is different from catalytic unambiguous classes mentioned in \cite{KMPS}). We are also 
able to  prove a tradeoff result between catalytic space and the workspace of the oracle, for $\NL$.
 Notably, our results are all in simultaneous polynomial time. 
They are, in fact, catalytic adaptations of the results of \cite{MP}. 
Historically, \cite{MP} is based on the results of \cite{KT,FGT} which yield
a partial derandomization of the isolation lemma in the context of reachability
and bipartite perfect matching respectively.
We are able to extend the results of \cite{MP} on $\LogCFL$ to the catalytic 
setting as well.

The results are based on derandomizing the isolation lemma for reachability 
and semi-unbounded circuit evaluation problem, respectively for $\NL,\LogCFL$. A brief description of our strategy is given below.

Given a layered directed acyclic graph $G$, \cite{MP} gives a procedure to give vertex weights so as to isolate min unique $s-t$ paths for all $s,t\in V(G)$. They use several logspace hash functions, where the seed length is logarithmic, to construct these weight assignments. When \cite{MP} require a new hash function, they iterate over all hash functions to find a `good' hash that works for their purposes. Once they have found a `good' hash, they store it, and either have found enough hash functions, or move on to find the next hash function. Once the minimum weight path has been isolated, reachability can be computed in unambiguous logspace (if the weights are larger, then in logarithmic space in the magnitude of weights), as shown by \cite{AR}. The same techniques have been adapted for the semi-unbounded circuit evaluation problem for logarithmic depth circuits in \cite{AR} and \cite{MP}.

We interpret the catalytic tape contents as hash functions. When we need a new hash function, we assume the corresponding block of catalytic tape to be the hash function. We check if this hash function is indeed `good' or not using an unambiguous oracle. If it is, we can proceed following the algorithm in \cite{MP}. Otherwise, we observe that the number of `bad' hash functions is substantially smaller than the number of `good' hash functions. Hence, we remove the `bad' hash function from the tape, and replace it with its index among the set of all bad hash functions. In this process, we end up saving some space. If by the end of this process, we have found a sufficient number of hash functions, we construct the required weight assignments efficiently. Instead, if we have not been able to find the correct number of hashes, we have ended up saving enough space to run \cite{MP} directly. Later, we restore the catalytic tape contents to its initial configuration, again by replacing the indices of the bad hash functions with the corresponding hash functions.

In \cite{MP}, it is required to sample $\mcO(\log^{0.5}n)$ many hash functions, each hash requiring $\mcO(\log n)$ bits of space. In fact, a finer analysis of \cite{MP} shows that, for all $\alpha\in[0,0.5]$, using $\mcO(\log ^{1-\alpha}n)$ many hash functions, we can find isolating weight assignments such that the weight of every vertex is of bitlength $\mcO(\log^{1+\alpha}n)$. Using this idea, we get a tradeoff result, where we use $\mcO(\log^{2-\alpha}n)$ catalytic space, and an oracle machine with $\mcO(\log^{1+\alpha}n)$ workspace, while retaining the properties of the oracle machines i.e. unambiguity for $\NL$, and unambiguity with an extra stack for $\LogCFL$.

Interestingly, our result also shows that $\NL$ is in $\CSPACE[\log n,\log^2 n]$ if and only if $\UL$ is in $\CSPACE[\log n,\log^2 n]$. That is, it is no harder to prove a catalytic adaptation of Savitch's theorem for $\NL$ than it is for $\UL$.

To summarise:
\begin{result} (\Cref{thm:nl,thm:cir})
For all $\alpha \in [0,0.5]$:
\begin{itemize}
\item $\NL\subseteq \CTISP\left[\poly(n), \log n,\log^{2-\alpha}n\right]^{\U\TISP\left[\poly(n),\log^{1+\alpha}n\right]}$
\item $\LogCFL\subseteq \C\TISP[\poly(n),\log{n},\log^{2-\alpha}{n}]^{\U\AuxPDA\text{-}\TISP[\poly(n),\log^{1+\alpha}n]}$
\end{itemize}
\end{result}


\section{Preliminaries}
We recall standard definitions that will be used throughout.

\medskip
\noindent\textbf{Matchings and exact matchings.}
Let \(G=(V,E)\) be an undirected graph.
A \emph{matching} in \(G\) is a subset \(M\subseteq E\) such that no two edges in \(M\) share a common endpoint.
A matching \(M\) is \emph{perfect} if every vertex of \(G\) is incident to exactly one edge of \(M\).
More generally, given a  coloring of $V$ in two colors, \emph{red} and \emph{blue}, and an integer $k$, an \emph{exact matching} is a matching \(M\subseteq E\)
such that $M$ contains exactly $k$ \emph{red} edges.

\medskip
\noindent\textbf{Arborescences.}
Let \(D=(V,E)\) be a directed graph and let \(r\in V\) be a designated root.
An \emph{arborescence rooted at \(r\)} is a spanning subgraph \(T=(V,T)\subseteq D\) such that
(i) \(r\) has indegree \(0\),
(ii) every vertex \(v\neq r\) has indegree exactly \(1\), and
(iii) the underlying undirected graph of \(T\) is a tree.
Equivalently, an arborescence is a directed spanning tree whose edges are directed away from the root \(r\).

\medskip
\noindent\textbf{Linear Pfaffian Matroid Parity problem.}
Let $A$ be a matrix whose columns are partitioned into pairs, called \emph{lines}, denoted $L$. The problem asks to find a set of columns of $A$ that are linearly independent and is equal to a union of lines of $L$. Let $\mcB$ be the set of all such linearly independent sets of columns, each set in $\mcB$ is said to be a \emph{base}. For $J$ a subset of columns of $A$, we denote $A[J]$ as the submatrix of $A$ formed by the columns indexed by $J$. We say the matrix forms a  \emph{pfaffian parity} \cite{Oki} if $\det A[B]=c$ $\forall B\in\mcB$ for some constant $c$.

\medskip
\paragraph*{Standard complexity classes.}
\(\TC^1\) is the class of Boolean functions computable by uniform circuits of logarithmic depth and polynomial size with unbounded fan-in AND, OR, NOT and MAJORITY gates.

\medskip
\(\ZPP\) (Zero-error Probabilistic Polynomial time) is the class of decision problems
for which there exists a probabilistic polynomial-time Turing machine \(M\) satisfying:
for every input \(x\),
\begin{enumerate}
	\item \(M(x)\) outputs either the correct answer or ``don't know'', and
	\item \(\Pr[M(x)\text{ outputs ``don't know''}] \le 1/3.\)
\end{enumerate}
Equivalently, $\ZPP=\RP\cap \coRP$.

$\NL$ is the set of languages decidable by a nodeterministic logspace machine, and a $\LogCFL$ machine is a $\NL$ machine with an additional auxilary stack. $\SAC^1$ is the class of boolean circuits with unbounded fan-in OR gates, fan-in two AND gates and NOT gates of logarithmic depth and polynomial size. It is well known that $\LogCFL=\SAC^1$. It suffices to assume that all the NOT gates are on the bottom level, and are absorbed in the literals, making the circuit NOT free. We say that a circuit is layered if the edges in the circuit are only between consequitive layers, and gates at odd and even layers consist of $\wedge$ and $\vee$ gates respectively.  For a $\SAC^1$ circuit $C$, define a \emph{proof tree} $F$ as follows: $(1)$ output gate $g_{out}$ is in $F$, $(2)$ for all $\wedge$ gates $g\in F$, both children of $g$ are in $F$, $(3)$ for all $\vee$ gates $g\in F$, exactly one child of $g$ is in $F$, and $(4)$ all the literals in $F$ are set to \emph{true}. 
\begin{proposition}\label{nlhard}
	Given a layered directed acyclic graph $G=(V,E)$ on $n$ vertices and two special vertices $s$ and $t$ in $V$, checking whether $t$ is reachable from $s$ in $G$ is $\NL$ hard under logspace reductions.
\end{proposition}
\begin{proposition}\label{sachard}
	Given a layered $\SAC^1$ circuit $C$, checking for the existence of a proof tree of $C$ is $\LogCFL$ hard under logspace reductions.
\end{proposition}

$\UTISP[t(n),s(n)]$ is the class of languages accepted by a non-deterministic turing machine running in time $t(n)$ and space $s(n)$, with the further restriction that for any input $x$, the non-deterministic turing machine has at most one accepting path.
A $\Unm\AuxPDA$-$\TISP\left[t(n),s(n)\right]$ machine is a non-deterministic machine that runs in time $t(n)$ and space $s(n)$ that is augmented with an auxilary pushdown stack separate from its work space, such that the non-deterministic machine has at most one accepting path, i.e. it is an $\UTISP[t(n),s(n)]$ machine with an auxilary stack.

{$\MA$:} A language $\mcL$ is in $\MA$ if Merlin sends a polynomial-size witness and Arthur verifies it probabilistically in polynomial time. \\
Formally, there exists a PPT verifier $V$ such that  
$x\in \mcL \Rightarrow \exists w:\Pr[V(x,w)=1]\ge 2/3$ and  
$x\notin \mcL \Rightarrow \forall w:\Pr[V(x,w)=1]\le 1/3$.

{$\S_2\P$:} A language $\mcL$ is in $\S_2\P$ if there exists a poly-time predicate $V(\cdot,\cdot,\cdot)$ such that  
$x\in \mcL \Rightarrow \exists y\,\forall z\, V(x,y,z)=1$ and  
$x\notin \mcL \Rightarrow \exists z\,\forall y\, V(x,y,z)=0$.

\medskip
\paragraph*{Catalytic Computation}
\begin{definition}
	A catalytic Turing machine $M$ with workspace bound $s(n)$ and catalytic space bound $c(n)$ is a Turing machine that has a read-only input tape of length $n$, write-only output tape (with tape-head that moves only left to right), an $s(n)$ space-bounded read-write work tape, and a catalytic tape of size $c(n)$. 
	We say that $M$ computes a function $f$ if for every $x \in\{0,1\}^n$ and $\tau \in\{0,1\}^{c(n)}$, the result of executing $M$ on input $x$ with initial catalytic tape $\tau$ i.e. $M(x,\tau)$ satisfies:
	\begin{enumerate}
		\item $M$ halts with $f(x)$ on the output tape.
		\item $M$ halts with the catalytic tape consisting of $\tau$.
	\end{enumerate} 
	$\CSPACE[s(n),c(n)]$ is the family of functions computable by such a Turing machine, and $\CTISP[t(n),s(n),c(n)]$ is the family of functions computable by such a Turing machine that simultaneously runs in time $t(n)$.
\end{definition}
We define \emph{catalytic logspace} as the class 
\[\CL=\cup_{k\in\bbN}\CSPACE[k\log n,n^k]\]
and $\CL\P$ is the set of languages decidable by a $\CL$ machine that simultaneously runs in polynomial time.
\begin{lemma}[\cite{BCKLS14}]
	$\TC^1\subseteq \CL$
\end{lemma}
	

Oracle $\CL$ machines are defined analogous to logspace oracle machines.  That is to say, an oracle $\CL$ machine $M$ is equipped with an oracle $A\subseteq \Sigma^*$.  In addition to the input tape, catalytic tape and worktape, it also has a query tape on which it can write an oracle query $q$ and enter into a special query state $Q$. Depending on whether or not $q\in A$ the next state is $Q_Y$ or $Q_N$ respectively. The important restriction is that the query tape head is allowed to move only from left to right as it writes the query $q$, and the query tape must be reset after a query has been made.

We will also consider oracle $\CL$ machines that can make $k$ rounds of \emph{parallel queries} to the oracle $\NP$, namely, $\CL^{\NP}_{k\text{-}round}$ for $k\in \bbN$. We shall define this class in \Cref{sec:oracle}, see \Cref{def:or}.

	
		

Lastly, we also look at the oracled classes $\C\TISP[\poly(n),\log n,\log^{2-\alpha} n]^A$ where the oracle $A$ is $\U\TISP[\poly (n),\log^{1+\alpha}n]$ and $\U\AuxPDA\text{-}\TISP[\poly(n),\log^{1+\alpha}n]$ for $\alpha\in[0,0.5]$. We show that $\NL$ and $\LogCFL$ are contained in these classes respectively.

\section{Search to Decision}\label{sec:three}
\paragraph*{Decision vs Search Problems.}

In this section we formulate the problem of reducing search to
decision in a fairly general setting and show that the isolation lemma
of \cite{MVV} can be used to prove that search is reducible to a
weighted version of the corresponding decision problem in catalytic logspace. 
More precisely, we
show that there is a catalytic logspace oracle machine that solves the
search problem with parallel queries to the weighted decision problem
as oracle. It turns out that this general formulation and its solution
yields a number of interesting concrete applications.

Let $\mcL \subseteq \{0,1\}^*$ be a language (equivalently, \emph{decision
problem}) defined by a binary predicate $R_\mcL$ in the following
sense:
\[
x \in \mcL \quad \Longleftrightarrow \quad \exists\, y \in \{0,1\}^{m(|x|)} 
\text{ such that } R_\mcL(x,y)=1,
\]
where $m(|x|)$ is polynomial in $|x|$. If $x\in\mcL$ we think of a
$y\in\{0,1\}^{m(|x|)}$ as a solution for $x$ that witnesses
$x\in\mcL$. Computing such a $y$ for $x\in\mcL$ is the corresponding
\emph{search problem}. 
\begin{remark}
Notice that this is similar to the definition of $\NP$.
However, we do not insist that $R_\mcL$ is polynomial-time
computable. In this sense, it is more general.
\end{remark}


\paragraph*{Weighted Decision.}
Suppose $\mcL$ is a decision problem defined by predicate $R_\mcL$
with witness length $m(\cdot)$ as defined above. For an input
$x\in\Sigma^*$, let $W : [m] \to \mathbb{N}$ be a weight assignment to
the $m=m(|x|)$ positions of a candidate witness string
$y\in\{0,1\}^m$, where we define the weight of $y$ as: 
\[
W(y) = \sum_{i : y_i = 1} W(i).
\]
We define the \emph{weighted decision problem} as the language $\mcL_w$:

\[
  (x,W,w_0)\in\mcL_w
\quad \Longleftrightarrow \quad
\exists\, y \text{ such that } R_\mcL(x,y)=1 \text{ and } W(y)\le w_0.
\] 


Furthermore, \emph{min-weight search problem}
is defined as follows:
given $x\in\Sigma^*$, compute a witness $y \in \{0,1\}^{m(|x|)}$ of minimum weight
$W(y)$ such that $R_\mcL(x,y)=1$.
That is, the problem is to output
\[
y^\ast \in \arg\min_{y : R_\mcL(x,y)=1} W(y).
\]
We define min-weight-$\search_\mcL(x,W)$ as  the set of
all min-weight witnesses i.e $\arg\min_{y : R_\mcL(x,y)=1} W(y)$.

\begin{remark}
  It turns out that the only constraint we need on $R_\mcL$ to show
  the $\CL$-computable reduction in this section is that it is in
  $\PSPACE$. Throughout this section, we consider languages $L$ such
  that $R_{\mcL}$ is in $\PSPACE$, i.e given $x,y$, there is a $\PSPACE$
  algorithm to check whether $R_\mcL(x,y)=1$ or not. Henceforth, we
  drop the subscript $\mcL$ when it is clear from the context.
\end{remark}

\begin{definition}
  Let $E=\{e_1,e_2,\ldots,e_n\}$ be a universe of $n$ elements,
  $\mcF\subseteq 2^E$, and $W$ a weight assignment $W:E\to \bbN$. We
  say that $W$ \emph{isolates a min-weight set} in $\mcF$ (briefly, $W$
  \emph{min-isolates} $\mcF$), if there exists a unique min-weight set $F$ in
  $\mcF$ i.e.  $\arg\min_{F\in\mcF}W(F)$ is a single set. Moreover,
  if $W$ does not min-isolate $\mcF$, then there exists
  $F_1\ne F_2\in\mcF$ where $W(F_1)=W(F_2)$ is the minimum weight. In
  such a case, any element $e\in F_1\setminus F_2$ is a threshold element. 
\end{definition}

Let $\mcL$ be a language, and $x$ an input such that $|x|=n$ and
$m(n)=m$. Let $E=[m]=\{1,\ldots,m\}$. For $y\in\{0,1\}^m$, let
$F_y=\{i:y_i=1\}$, which defines a natural bijection between witness
candidates $y$ and subsets of $E$. We will sometimes use $F_y$ to
denote a witness candidate $y$.  In the sequel, we only consider
weight functions $W:E\to [0,m^c]$, where $[0,m^c]$ denotes the
integers in this range, for some constant $c>0$ that is independent of
$n$. We will refer to such $W$ as \emph{polynomially bounded} weight
functions.

Given an isolating weight assignment $\bmW$ such that there is a
unique minimum weight solution $F_y\subseteq E$, we can search for
$F_y$ in logspace, with queries to the oracle $\mcL_w$ as follows:
\begin{enumerate}
\item Let $W^{min}$ be the least weight in $[0,m^c]$ such that
  $(x,\bmW,W^{min})\in \mcL_w$. That is,
  $W^{min}=\min_{F : R_\mcL(x,F)=1} W(F)$. Then $W^{min}$ can be found by
  binary searching for the least $w_0$ such that
  $(x,\bmW,w_0)\in \mcL_w$.
\item For each $e\in E$, define weight assignment $\bmW_e$ which is
  the same as $W$ everywhere except at $e$, where $W_e(e)=\infty$ (we
  can replace $\infty$ with $W^{min}+1$). Then $e\in F$ \emph{iff}
  $(x, W_e,W^{min})\notin \mcL_w$.
\end{enumerate}

The goal is, therefore, to find an isolating weight assignment. We
assume that the catalytic tape contains polynomially many weight
assignments. If we find an isolating weight assignment, we are
done. Otherwise, we find a threshold element and efficiently compress
the weight function to save some space.  Ultimately, we save enough
space to brute force over all subsets to find a correct set $F$.

More precisely, we next show that there is a catalytic logspace
procedure that takes input $x$, makes queries to the weighted decision
problem $\mcL_w$, and computes a $y\in\{0,1\}^{m(|x|)}$ such that
$R_\mcL(x,y)$ (i.e.\ search is $\le^\CL_T$-reducible to the weighted
decision).

\begin{theorem}\label{thm:srch}
  For polynomially bounded weights, min-weight-$\search\le^{\CL}_T$
  weighted-$\decision$ for any language $\mcL$ defined by a
witness relation $R_\mcL(x,y)$. Equivalently, the search problem 
min-weight-$\search$ is
  computable in $\CL^{\mcL_w}$. In particular, it holds for $\NP$-languages.
\end{theorem}

      
\begin{proof}
  We describe the claimed oracle $\CL$ procedure, \Cref{alg:search},
  and prove its correctness. We are given $E=[m]$ and the catalytic
  tape $\mcC$ containing $N$ weight functions
  $\mcC=(\tau_1=\bmW_1,\tau_2=\bmW_2,\ldots,\tau_N=\bmW_N)$, where
  each $\bmW_i:E\to [0,m^2]$ is a weight assignment, and $N=\poly(n)$
  shall be specified later. The weight function $\bmW_i$ is encoded as
  a $2m\log m$ bit string $(W_i(1),\ldots,W_i(m))$. Clearly, if
  $(x,\bmW_1,m^3)\notin \mcL_w$ then we can reject the input because
  the weight of all subsets of $[m]$ are bounded by $m^3$. We can
  iterate over all the weight assignments $\bmW_1,\ldots,\bmW_N$ (or
  query them in parallel). Now, suppose we are processing $\bmW_i$, we
  calculate $W_i^{min}$ as the least $w_0$ such that
  $(x,\bmW_i,w_0)\in \mcL_w$.

	
\begin{claim}\label{claim:1}
  If $\bmW_i$ is not isolating then there is a threshold element
  $ e\in E$ such that $(x,\bmW_i',W_i^{min})\in \mcL_w$ and
  $(x,\bmW_i'',W_i^{min}-W_i(e))\in \mcL_w$, where
  weight functions $\bmW_i'$ and $\bmW_i''$ are the same as $\bmW_i$
  everywhere except $e$ and $W_i'(e)=W_i^{min}+1$, $W_i''(e)=0$, as
  defined in \Cref{alg:search}.
\end{claim}

\begin{proof}
  Suppose $\bmW_i$ is not min-isolating. Then there are
  $F_y\neq F_z\subseteq E$ that are both minimum weight witnesses for
  $\bmW_i$. For $e\in F_y\setminus F_z$ consider $\bmW_i',\bmW_i''$ as
  defined. Clearly, $W_i''(F_y)=W_i(F_y)-W_i(e)=W_i^{min}-W_i(e)$ and
  $W_i'(F_z)=W^{min}$ and $R_\mcL(x,y)=R_\mcL(x,z)=1$. Thus, we have
  $(x,\bmW_i',W_i^{min})\in \mcL_w$ and
  $(x,\bmW_i'',W_i^{min}-W_i(e))\in \mcL_w$.
\end{proof}
Now, we consider the following cases.
\begin{enumerate}
\item $\bmW_i$ is min-isolating. In this case we can find in logspace
  the unique min-weight set $F_y$ with queries to the decision oracle
  as described in the sketch above. Recomputing 
  $\tau_1=\bmW_1,\ldots,\tau_{i-1}=\bmW_{i-1}$ is as described in the
  third step below (\cref{alg:recomp}).

\item $\bmW_i$ is not min-isolating. Then find an index $i_0$
  satisfying the above claim (i.e. $e=i_0$) with queries to the
  decision oracle. Replace $\tau_i$ with $(i_0,\bmW_i^{i_0})$ where
  \[
    \bmW_i^{i_0}=(W_i(1),\ldots,W_i({i_0-1}),W_i({i_0+1}),\ldots,W_i(m))
   \] 
  is the $m-1$ weight vector with the $i_0^{th}$ entry dropped. This
  saves $\log m$ bits of space since $i_0$ requires $\log m$ space and
  $W_i({i_0})$ requires $2\log m$ bits to store. Next, we move on to
  process $W_{i+1}$.
		
  If none of $W_1,\ldots,W_N$ are min-isolating, we have freed
  $N\log m$ bits of space in the catalytic tape. We can search for a
  min-weight witnessing set $F$ as follows\footnote{First shift all
    the free space to the right making it a contiguous block of
    space. Then, after completing computation, shift it back.}:
		
  Iterate over all subsets of $F_y\subseteq E$.
		
  Compute each $R_{\mcL}(x,y)$ in $\PSPACE $ and return the least weighted
  $F_y$ for which $R_{\mcL}(x,y)=1$. Suppose computing $R_{\mcL}(x,y)$ requires
  $n^c$ space. Then this procedure requires $m+n^c$ space ($m$ since
  we iterate over all subsets of $[m]$). So choosing $N=m+n^c$
  suffices.
				
  Next we have to recompute all the $\tau_1,\ldots,\tau_N$.
		
\item Recompute (\cref{alg:recomp}): Suppose we are recomputing
  $\tau_i=W_i$ such that $W_i$ is not min-isolating, given
  $\tau_i'=(i_0,W_i^{i_0})$. Let $\bmW_i'$ and $\bmW_i''$ be same as
  $\bmW_i$ everywhere except $i_0$ and $W_i'(i_0)=m^3$,
  $W_i''(i_0)=0$. Since, $W_i$ is not min-isolating and $i_0$ is a
  threshold element, we know that $W_i^{min}=W_i'^{min}$ which we can
  compute using the $\decision$ oracle. Moreover $W_i''^{min}$ can
  also be computed. By the above claim, we can set
  $W_i({i_0})=W_i'^{min}-W_i''^{min}$, and reset $\tau_i$ to
  $(W_i(1),\ldots,W_i(m))$.
\end{enumerate}
	
Clearly, the procedure described above is in $\CL$, and the catalytic
tape contents are always restored to its initial configuration. 	
\end{proof}

\begin{remark}
  We can show a similar $\CL$-computable Turing reduction for decision
  problems that are already weighted, the instance $x$ comes with an
  input weight function $W_{input}$ that is polynomially bounded in
  $|x|$. In order to adapt the above proof it suffices to modify the
  weight functions (which will isolate a minimum weight solution) as
  follows: At step $i$, we consider weight $W_{input}\cdot m^{10}+W_i$
  where $W_i$ is the Catalytic weight. Further, in the proof, we define a weight assignment $W'$ that agrees everywhere except a threshold element $e$ where $W'(e)=m^3$ (during the Recompute procedure, \cref{alg:recomp}). In the case where we are already given the input weights $W_{input}$ bounded by $m^c$, we replace $W'(e)=m^{c+20}$ (i.e. some large enough polynomially bounded value). The rest of the proof remains the same.
  
\end{remark}

We show as
corollary that if $\mcL_w$ for polynomially bounded weights is $\CL$
computable then the search problem min-weight-$\search$ is also $\CL$
computable. By Theorem~\ref{thm:srch} we have that
min-weight-$\search$ is in $\CL^{\mcL_w}$. Now suppose that $\mcL_w$
is in $\CL$. We cannot directly conclude from Theorem~\ref{thm:srch}
that min-weight-$\search$ is in $\CL$, because we do not know if
$\CL^\CL=\CL$. The difficulty here, unlike showing say
$\Log^\Log=\Log$, is in simulating the $\CL$ machine for the oracle
because \emph{its input} is written on the query tape of the base
$\CL$ machine. More precisely, suppose $M$ is a $\CL^A$ machine
computing some function $f$, where the oracle language $A$ is computed
by $\CL$ machine $M_A$. For a single $\CL$ machine $M'$ to compute
$f$, $M'$ needs to simulate $M$ and, when $M$ queries $A$ for a string
$q$ it needs to simulate $M_A$ on $q$. The problem here is that $q$
cannot be held in the worktape and $M'$ needs to simulate $M$ to
access bits of $q$ multiple times. It is not clear if this is possible
because $M$ can change its catalytic tape while computing the query
$q$.  However, if we consider a more restrictive model in which the
base $\CL$ $M$ \emph{does not} change its catalytic tape when it
writes a query $q$, then it is indeed true that $f$ can be computed in
$\CL$. The $\CL^{\mcL_w}$ computation that is described in the proof
of Theorem~\ref{thm:srch} is precisely of this form. The queries made
to the $\mcL_w$ oracle are of the form $(x,W,w_0)$, and we can observe
that this is available in the memory of the base $\CL$ machine. More
precisely, $x$ is on the input tape, $W$ is on the catalytic tape
(with a minor $\mcO(\log n)$ bit modification kept on the work tape) 
and $w_0$ is available on
the worktape. Thus, the $\CL$ machine that solves $\mcL_w$ can be
simulated (with a separate catalytic tape) for each query made by the
base $\CL$ machine, preserving the catalytic property. 



\paragraph*{Some $\CL$ Search Algorithms based on Theorem~\ref{thm:srch}}

Here we state some direct consequences of \Cref{thm:srch}, we assume
that the weights are polynomially bounded.

\begin{corollary}\label{cor:srchalgos}
  The following search problems are $\CL$ computable:
\begin{enumerate}
\item The perfect matching search problem for planar graphs. 
\item The red-blue matching search problem for planar graphs.
\item The minimum-weight $r$-arborescence problem in polynomially
  weighted digraphs.
\end{enumerate}  
\end{corollary}  

\begin{proof}
  \begin{enumerate}
  \item Kastelyn \cite{Kast} showed that counting matchings of planar
    graphs is in deterministic polynomial time by efficiently computing
    a Pfaffian orientation for the given planar graph. It is shown in
    \cite{Vazirani} that the same can be done for $K_{3,3}$-free graphs as
    well. Mahajan et al \cite{MSV} showed counting perfect matchings
    in planar graphs is even $\GapL$ computable and hence is in
    $\TC^1$. For polynomially weighted graphs, the idea of encoding a
    nonnegative integer weight $w_e$ of an edge $e$ of a planar graph
    as a univariate monomial $y^{w_e}$ allows us to compute the
    Pfaffian as a small degree univariate polynomial in $y$, where the
    coefficient of $y^{w_0}$ is the number of perfect matchings of
    weight $w_0$, which can be computed in $\TC^1$. It follows that
    the weighted decision problem can be solved in $\TC^1$ and hence
    $\CL$. Hence, \Cref{thm:srch} yields the claimed $\CL$
    search algorithm.
  \item Next, we consider the red-blue perfect matching problem for
    planar graphs. This is also known as the exact perfect matching
    problem. The input graph has edges colored red and blue and the
    problem is to find a perfect matching with exactly $k$ red edges
    (for a $k$ given with the input). This problem has a randomized
    $\NC$ algorithm even for general graphs but is not known to be in
    $\P$ (for general graphs). However, using \Cref{thm:srch} we
    can again show that in planar graphs searching for a red-blue
    perfect matching with exactly $k$ is also in $\CL$. This follows
    because the corresponding weighted decision problem (is there a
    perfect matching with exactly $k$ red edges of weight at most
    $w_0$ for a given weight assignment $W$) can be shown to be in
    $\TC^1$ as the corresponding counting problem is in $\TC^1$
    \cite{Yuster,MSV,AJMV}. This upper bound of $\CL$ is intriguing as the
    search version of exact red-blue perfect matchings in planar
    graphs is not even known to be in deterministic $\qNC$.


  \item We now consider the problem of searching for a minimum weight
    rooted $r$-arborescence in an input weighted directed graph (with
    polynomially bounded weights). We can use the \emph{directed
      matrix-tree theorem} (see, for example, \cite{Zeilberger})
    to count the number of weighted
    arborescences rooted at $r$ in such directed graphs (by the same
    trick of encoding the weight $w_e$ of each edge $e$ as the
    monomial $y^{w_e}$). From the resulting determinant we can read
    off the number of weighted $r$-arobrescences for each weight $w$
    as the coefficient of $y^w$. Now, as explained in the remark after
    Theorem~\ref{thm:srch}, we can consider additional weights to make
    the reduction of Theorem~\ref{thm:srch} and then
    \Cref{thm:srch} applicable. First, as the weighted
    decision problem requires only determinant computation, it is in
    $\TC^1$ and hence in $\CL$. Consequently, the search problem is
    also in $\CL$.




\end{enumerate}
\end{proof}  


\begin{remark}
  The upper bounds shown in Corollary~\ref{cor:srchalgos} are in fact
  $\CL\P$, since in our reduction, if we have saved enough space, we
  can run a polynomial-time algorithm to search for the corresponding
  witness in the above cases. The $\CL\P$ upper bound for these also follows directly since $\CL\P=\CL\cap\P$ \cite{CLMP}.

  In \cite{Oki}, a number of interesting algorithmic problems based on
  linear matroid parity are studied. We can see that the problem of
  counting weighted bases is in $\TC^1$ if the matrix representing the
  given matroid is a \emph{pfaffian parity}. It is shown in \cite{Oki}
  that some natural problems can be modeled as base search for linear
  matroid parity problem (with the pfaffian parity
  property). \Cref{thm:srch} will imply a $\CL$ upper bound
  for these problems as well.
\end{remark}


\begin{algorithm}
	\caption{SearchtoWeightedDecision($x$)}\label{alg:search}
	\begin{algorithmic}[1]
		\State \textbf{Input:} $E=[m],\mcC=(\tau_1=\bmW_1,\tau_2=\bmW_2,\ldots,\tau_N=\bmW_N)$ where $\bmW_i=(W_i(1),\ldots,W_i(m))$ is a weight assignment and each $W_i(j)$ is $\le n^2$ i.e. requires $2\log n$ bits to store.
		\State \textbf{Output:} If $x\in\mcL $, outputs a $F\in\search(E)$
		\If{$(x,\bmW_1,m^3)\notin \mcL_w$}
		\State Reject
		\EndIf
		\State $k\gets0$
		\For {$i\in [N]$}
		\State $k\gets i$
		\State Compute $W^{min}_i$ by searching for the least $w_0$ such that $(x,\bmW_i,w_0)=\mcL_w$.
		
		\State $flag\gets 1$, $index\gets 1$
		\While {$flag=1$}
		\State  $\bmW_i'$, $\bmW_i''$ are the same as $\bmW_i$ except at $index$ and $W_i'(index)\gets W_i^{min}+1$, $W_i''(index)\gets 0$
		\If {$(x,\bmW_i',W_i^{min})\in\mcL_w$ and $(x,\bmW_i'',W_i^{min}-W_i(index))\in\mcL_w$}
		\State $flag=0$
		\State break
		\EndIf
		\State $index\gets index+1$
		\EndWhile
		\If {$flag=1$}\Comment{$W_i$ is a min-isolating weight}
		\For {$e\in E$}
		\State Define $\bmW_i',\bmW_i''$ w.r.t. $e$ as before.
		\If {$(x,\bmW_i',W_i^{min})\in\mcL_w$}
		\State \Return $(e\in F)$ \Comment{$F$ be the min-unique set in $\search(E,W_i)$}
		\EndIf
		\EndFor
		\State $k\gets k-1$
		\State break
		\Else \Comment{$\bmW_i$ is not isolating}
		\State $\bmW_i^{i_0}$ is $(W_i(1),\ldots,W_i({i_0-1}),W_i({i_0+1}),\ldots,W_i(m))$ \Comment{where $i_0=index$}
		\State Replace $\tau_i$ with $\tau_i'\gets (index,\bmW_i^{index})$. 
		\EndIf
		\If{$i=N$}
		\State Use free space to brute force.
		\EndIf
		\EndFor
		\For {$i$ from $1$ to $k$}
		\State Recompute($\tau_i'$)\Comment{Call \Cref{alg:recomp}}
		\EndFor
		
	\end{algorithmic}
\end{algorithm}
\begin{algorithm}
	\caption{Recompute($\tau_i'$)}\label{alg:recomp}
	\begin{algorithmic}[1]
		\State \textbf{Input:} $E=[m],\tau_i'=(index,\bmW_i^{index})$ where $\bmW_i^{index}=(W_i(1),\ldots,W_i(m))$ except the $index$th position is absent.
		\State \textbf{Output:} $\bmW_i= (W_i(e_1),\ldots,W_i(e_n))$ i.e. the initial configuration of $\tau_i$.
		\State Let $i_0\gets index$, $\bmW_i'$ and $\bmW_i''$ are the same as $\bmW_i$ everywhere except $i_0$ and $W_i'({i_0})\gets m^3$, $W_i''({i_0})\gets 0$
		\State Compute $W'^{min}_i$ by searching for the least $w_0$ such that $(x,\bmW_i',w_0)\in\mcL_w$.
		\State Compute $W''^{min}_i$ by searching for the least $w_0$ such that $(x,\bmW_i'',w_0)\in\mcL_w$.
		\State Define $W_i({i_0})\gets W_i'^{min}-W_i''^{min}$
		\State $\tau_i=(W_i(w_1),\ldots,W_i(e_n))$
	\end{algorithmic}
\end{algorithm}

\section{$\CL$ with $\NP$ oracle}\label{sec:oracle}
In this section, we explore catalytic logspace computation augmented with an $\NP$ oracle. We first observe that the containment $\CL\subseteq\ZPP$ relativizes (we note that this is mentioned without
proof in \cite{CGMPS}. The proof is, in essence, very similar to that of the containment $\CL\subseteq\ZPP$ \cite{BCKLS14}.

\begin{lemma}\label{lemma:clzp}
  For any oracle $A\subseteq \Sigma^*$, we have that $\CL^A\subseteq \ZPP^A$.  \end{lemma}

\begin{proof}
  Let \(M\) be a deterministic $\CL$ oracle machine with oracle access to $A$.  Consider an input \(x\) of length \(n\), and let \(p(n)\) be the polynomial bound on the size of the catalytic tape of \(M\). For each initial catalytic string \(\tau\in\{0,1\}^{p(n)}\) write \(M(x,\tau)^A\) for the computation of \(M^A\) on input \(x\) with auxiliary start-string \(\tau\). By definition, \(M(x,\tau)\) begins in the configuration \((\mathrm{start},\tau)\) and halts in \((\mathrm{acc},\tau)\) or \((\mathrm{rej},\tau)\); in particular the auxiliary tape is restored to \(\tau\) on termination. Let \(\ell(\tau)\) denote the number of steps of the run \(M(x,\tau)^A\).
	
		\medskip
	
	\noindent\textbf{Claim.} Let $\tau_1,\tau_2,\tau_3,\tau_4$ be any four distinct catalytic tape
contents, and\footnote{Here, by a configuration $(u,\tau)$ we mean
		$u$ to be the logspace worktape contents, where $u$ also includes the head positions on the input tape, worktape and catalytic tape, and $\tau$ denotes the catalytic tape contents.} 
$(u,\tau)$ be any given configuration. Then, $(u,\tau)$ cannot lie on the computation paths (that is, the sequence of configurations visited) of all four computations $M(x,\tau_1)^A,M(x,\tau_2)^A,M(x,\tau_3)^A$ and $M(x,\tau_4)^A$.
	
\noindent\emph{Proof of claim.} Suppose, for the sake of contradiction, that for four distinct
	\(\tau_1,\tau_2,\tau_3,\tau_4\), the four computations $M(x,\tau_i)^A$ for $i\in\{1,2,3,4\}$ reach the common configuration $(u,\tau)$. Let $M_i$ denote $M(x,\tau_i)^A$. For each $M_i$, we refer to the set of configurations reached between writing the first bit of a query and receiving the answer bit to be the corresponding query phase. Thus, for each $M_i$ there are as many query phases as
the number of queries made during that computation. Now, consider the following two exhaustive cases:
	\begin{enumerate}
		\item At configuration $(u,\tau)$ at least two of the $M_i$'s are not in a query phase. Suppose $M_1$ and $M_2$ are these runs of the machines. But then their subsequent computations after $(u,\tau)$ are entirely determined by $(u,\tau)$. Hence both $M(x,\tau_1)^A$ and $M(x,\tau_2)^A$ cannot end in the distinct states $(\mathrm{final},\tau_1)$ and $(\mathrm{final},\tau_2)$, contradicting the definition of catalytic computation. 
		\item At least three of the $M_i$'s are in a query phase at $(u,\tau)$, say $M_1, M_2,$ and $M_3$. Then these three computations therefore are in the same configuration $(u',\tau')$ right before they read the corresponding answer bits, say $b_1,b_2,b_3\in\{0,1\}$, to the respective query they each made to oracle $A$. But two of these bits have to be the same, say, $b=b_1=b_2$. But then the next configuration (and therefore the entire subsequent configurations) for the computation paths of $M_1$ and $M_2$ is determined by $(u',\tau')$ and $b$. This again contradicts the definition of catalytic computation, since both $M_1$ and $M_2$ cannot end at their corresponding initial catalytic states.
\end{enumerate}
 \(\square\)

	
	\medskip
	
	Observe that the total number of logspace worktape configurations (including all
three tape head positions) possible is $2^{\mcO(\log n)}$. 
Thus the total number of configurations of the form $(u_0,\tau_0)$ of $M$ is $2^{p(n)+\mcO(\log n)}$. Denote this quantity by $C$.
	
Let the computation \(M(x,\tau)^A\) have \(\ell(\tau)\) distinct configurations. From the above claim, we have that each configuration $(u,\tau)$ can be shared by at most three distinct computation paths (for distinct initial catalytic strings), and therefore:
	\[
	\sum_{\tau\in\{0,1\}^{p(n)}} \ell(\tau) \le 3C =3\cdot  2^{p(n)+\mcO(\log n)}.
	\]
	Dividing by \(2^{p(n)}\) yields the expected run length for a uniformly chosen
	\(\tau\):
	\[
	\mathbb{E}_{\tau}[\ell(\tau)] \le \frac{3\cdot2^{p(n)+\mcO(\log n)}}{2^{p(n)}} = 3\cdot 2^{\mcO(\log n)} = 3n^{\mcO(1)}.
	\]
	Suppose $\mathbb{E}_{\tau}[\ell(\tau)]\le 3n^c$.
	By Markov's inequality,
	\[
	\Pr_{\tau}\big[\ell(\tau) > 3n^{c+1}\big] \le \frac{\mathbb{E}[\ell(\tau)]}{3n^2}
	\le \frac{3n^c}{3n^{c+1}} = \frac{1}{n}.
	\]
	Consequently a uniformly random \(\tau\) yields a run that halts within \(3n^{c+1}\)
	steps with probability at least \(1-\tfrac{1}{n}\).
	
	Consider the following $\ZPP^A$ algorithm on input \(x\):
	pick \(\tau\in\{0,1\}^{p(n)}\) uniformly at random and simulate \(M(x,\tau)^A\)
	for at most \(3n^{c+1}\) steps. If the simulation halts within \(3n^{c+1}\) steps,
	output the run's accept/reject result. Otherwise output ``do not know/$\bot$''.
	Clearly, the above runs in polynomial time. Thus we have shown that $\CL^A\subseteq\ZPP^A$.
\end{proof}

\subsection*{$\CL$ with nonadaptive oracle access}

Now, we define and discuss the oracle class $\CL^A_{1\text{-}round}$
(and then the class $\CL^A_{k\text{-}round}$) for an oracle
$A\subseteq \Sigma^*$. That is, the class of languages accepted by an
oracle $\CL$ machine that can make \emph{one round} of polynomially
many queries \emph{in parallel} to the oracle $A$. The natural model
would work as follows: On input $x$, the $\CL$ oracle machine $M$
during its computation enters the query state; then it writes down, on
the one-way query tape, the queries $q_1,\ldots,q_m,
m=\poly(|x|)$. After writing down all the queries, it enters an
``answers'' state to obtain on the answer tape all the answer bits,
namely, $A(q_1),\ldots,A(q_m)$. Then $M$ can read these answer bits
once from left to right after which the answer tape is reset, and then
the machine resumes its subsequent computations. However, this
definition makes this oracle $\CL$ machine model as strong as
$\PSPACE$, even for a trivial oracle like $A=\{1\}$, and hence not
useful for our purpose.\footnote{To see this, suppose the catalytic
  tape contents are $\tau=\tau_1\ldots\tau_m$.  Then, $M$ can write
  down the $m$ queries $q_i=\tau_i$ on the query tape and the answer
  tape would contain $\tau$. Then $M$ can use the catalytic tape for
  any $\PSPACE$ computation and finally $M$ can restore the catalytic
  tape from the answer tape.}


We will define a natural restriction of the $\CL^A_{1\text{-}round}$
machine model, so that the class $\CL^A_{1\text{-}round}$ is contained
in $\ZPP^A$, and then we describe some interesting upper bound results
using this model with an $\NP$ oracle.


Suppose the $\CL^A_{1\text{-}round}$ machine $M$ enters the query
state $Q$ and makes the queries $q_1,\ldots,q_m$ on the query tape, 
and then exits the query state. Let us call this as the \emph{query 
phase} of the machine $M$. We place the restriction that during the
query phase the machine $M$ performs an $\FL$ computation, treating 
the input and catalytic tape as a read-only input tape and the query tape as 
the one-way write-only output tape using the $O(\log n)$ size worktape
for computation. Next, suppose the $m$ answer bits to the queries
written on the answer tape by the oracle are $A(q_1),A(q_2),\ldots,A(q_m)$. 
The machine has some configuration\footnote{Here, by a configuration $(u,\tau)$ we mean
  $u$ to be the logspace contents and $\tau$ the
  catalytic contents.} $(u_0,\tau_0)$ immediately after exiting query
state. Then, before $M$ reads the answer bit $A(q_1)$, $M$ is
restricted to can only perform the computation of a logspace
transducer $L_0$ on $(u_0,\tau_0)$ to obtain the next configuration
$(u_1,\tau_1)$. Moreover, the logspace transducer is further
restricted as follows: $L_0$ has a single input/output tape consisting
of $z=z_1\ldots z_p$ and a work tape of $O(\log p)$ bits where both
tapes are read-write (notice that there is no separate output
tape). To begin with, both the read and write heads are positioned on
the left end of the input/output tape. Their movements are constrained
as follows throughout the computation of $L_0$: The write-head can
move only to the right as it writes; The read-head can move in both
directions but never to the left of the write-head. That is, the
read-head can never read a tape symbol written by the write-head
(except for the last written symbol).


For instance, suppose $f\in \FL$ such that it maps
$\Sigma^p\to \Sigma^p$ for all $p$. Furthermore, if
$f(x)=f_1f_2\ldots f_p$ and $f_i$ depends only on the $(n-i)$-length
suffix $x_{i+1}x_{i+2}\ldots x_p$ of $x$, then $f$ can be computed by
this restricted $\FL$ model. We will call such an $\FL$ machine a
$\onew$-$\inplacefl$ machine. Now, returning to the computation of the
$\CL^A_{1\text{-}round}$ machine $M$: After $M$ performes $L_0$,
suppose it is in configuration $(u_1,\tau_1)$.  Next, it reads the
first answer bit $A(q_1)$ and it can again perform a
$\onew$-$\inplacefl$ computation $L_1$ to reach configuration
$(u_2,\tau_2)$ before it reads the next answer bit. Similarly, for
each subsequent answer bit $A(q_i)$ that it reads from the answer
tape, one at a time, it can perform a $\onew$-$\inplacefl$ computation
$L_i$ to obtain the next configuration
$(u_{i+1},\tau_{i+1})$. Finally, after all the answer bits are read,
$M$ resumes its $\CL$ catalytic computation. Moreover, when $M$ makes the query $q_i$ for $i\in[m-1]$, is is again allowed only a $\onew$-$\inplacefl$ computation $L_{i}'$ on its current configuration, before it makes the next query $q_{i+1}$.

\begin{remark}
  Our definition of $\onew$-$\inplacefl$ is motivated by
  the notion of $\inplacefl$ \cite{CGMPS}. In the definition of
  $\inplacefl$, the extra restriction we have imposed on the movements
  of the read/write heads of the input tape is dropped. Indeed, in
  \cite{CGMPS} it is shown that there are functions in $\inplacefl$
  that are not in $\FL$. However, it is easy to see that
  $\onew$-$\inplacefl\subseteq\FL$, as follows: Suppose $L$ is a
  $\onew$-$\inplacefl$ machine. Consider the
  $\FL$ machine $\hat L$ with same input tape as machine $L$, and a
  separate output tape. The read head of the input tape of $\hat L$ is
  the read head of the input tape of $L$, and the write head of the
  output tape of $\hat L$ is the write head of the input tape of $L$.
  Since the read head of the input tape of $L$ never reads an index
  that has already been written on, $\hat L$ can clearly simulate $L$.
\end{remark}

Now, we can formally define the class $\CL^A_{k\text{-}round}$ for
$k\in\bbN$, as follows.


\begin{definition}[nonadaptive oracle access]\label{def:or}
  The class $\CL^A_{1\text{-round}}$ is the set of languages $L$ accepted
  by $\CL^A_{1\text{-round}}$ machines that makes one round of
  non-adaptive queries to the class $A$. That is,
  when the $\CL^A_{1\text{-round}}$ machine for $L$ on input $x$
  enters a query state, it writes down the queries
  $q_1,\ldots,q_m, m=\poly(|x|)$ (by an $\FL$ computation, treating the input and catalytic 
tape as a read-only input tape and the query tape as the output tape), and 
gets access to the answers $A(q_1),\ldots,A(q_m)$ from $A$ which it can read 
once from left to right. Between reading $A(q_i)$ and $A(q_{i+1})$, for each
  $1\le i\le m-1$, and before reading $A(q_1)$, $M$ performs a $\onew$-$\inplacefl$ 
computation on its configuration (work tape, catalytic tape).
  We similarly define the class $\CL^A_{k\text{-round}}$ for each
  $k\in\mathbb{N}$, in which the $\CL^A_{k\text{-round}}$ machine accepting
  $L$ can make $k$ rounds of non-adaptive queries with the same
  restrictions.
\end{definition}

\begin{lemma}\label{lem:cltoz}
$\CL^A_{k\text{-}round}\subseteq \ZPP^A$ for any oracle $A$. 
\end{lemma}  

\begin{proof}
  Recall from \Cref{lemma:clzp} that for any oracle $A$,
  $\CL^A\subseteq \ZPP^A$. It is also not hard to see from its proof
  that $\CL^A_{k\text{-}round}\subseteq \ZPP^A$. We sketch the proof
  here, while pointing out the key differences from \Cref{lemma:clzp}.
  Let $M$ be a $\CL^A_{k\text{-}round}$ machine with input $x$ and
  catalytic configuration $\tau$. Suppose the computaion path of
  $M(x,\tau)^A_{k\text{-}round}$ reaches $(u_0,\tau_0)$ when it starts to make
  the first round of queries (i.e. at this configuration, it enters the query state), and $(u_1',\tau_1')$ is the
  configuration after $M$ has processed all the answer bits. Similarly
  $(u_i,\tau_i)$ is the confuguration when $M$ starts to make the $i+1^{th}$
  round of queries, and $(u_{i+1}',\tau_{i+1}')$ is the configuration
  after $M$ has processed all the answer bits for this round of queries. Then consider the set
  of configurations of $M(x,\tau)^A_{k\text{-}round}$ as all the
  configurations in the computation path except the configurations
  between $(u_i,\tau_i)$ and $(u_{i+1}',\tau_{i+1}')$ for all
  $i\le k$. Then w.r.t this definition, the set of configurations of
  $M$'s computation path for distinct initial catalytic configurations
  $\tau$ and $\tau'$ will be disjoint (follows from case $1$ of the Claim in the proof of \Cref{lemma:clzp}). Moreover, the computations done
  during making the queries and processing the answer bits of a query round are all in $\FL$
  and thus in polynomial time. Thus, we will again have that the
  expected length of a computation path of the machine
  $M^A_{k\text{-}round}$ is polynomially bounded. Therefore, we can
  again simulate $M$ using a $\ZPP^A$ machine by randomly chosing an
  initial catalytic tape configuration and simulating $M$. Moreover,
  this $\ZPP^A$ machine can ask each round of queries of $M$ one by one
  and store the answers, since we do not have any space restrictions
  on $\ZPP^A$ machines. This shows that
  $\CL^A_{k\text{-}round}\subseteq \ZPP^A$.
\end{proof}

\begin{remark}
	\begin{enumerate}
		\hfil{}
		\item We can also consider the stronger model of $\CL^A_{k\text{-}round}$ as follows: When the machine $M$ is processing an answer phase after making a round of queries, it is allowed to make an $\inplacefl$ computation on its current configuration $(u,\tau)$ to get the next configuration $(u',\tau')$. It is also easy to see that \Cref{lem:cltoz} holds for this stronger model as well.

 \item  We notice that for $\NP$ oracles we even have the inclusion
  $\CL^\NP_{k\text{-}round}\subseteq \ZPP^{\NP[\mcO(k\cdot \log
    n)]}$. This is because one round of $\NP$ queries made by the
  $\CL$ machine can be simulated by a $\ZPP^{\NP[\mcO(\log n)]}$
  computation that can use the $\NP$ oracle to do a binary search for
  the number of $\SAT$ queries that are satisfiable.
  	\end{enumerate}
\end{remark}



Observe that \Cref{thm:srch} already implies that $\fF\in \CL^{\NP}$ since the weighted $\decision$ oracle for $\SAT$ can be simulated using an $\NP$ oracle.  Next, we show that the Valiant-Vazirani \cite{VV} reduction, showing that $\SAT$ is randomized polynomial-time reducible to unique-$\SAT$, can be derandomized in $\CL$, with nonadaptive queries to an $\NP$ oracle.

\begin{lemma}
	There is a $\cltonp$ algorithm $\mcA$ that takes as input  a boolean formula $\phi$, and returns another boolean formula of the form $\mcA(\phi)=\phi\wedge\phi'$ such that
	\begin{itemize}
		\item if $\phi$ is satisfiable, then $\mcA(\phi)$ has exactly one satisfying assignment.
		\item if $\phi$ is not satisfiable, then so is $\mcA(\phi)$.
	\end{itemize}
\end{lemma}
\begin{proof}
	Let $\phi$ be the input formula on the literals $X=x_1,x_2,\ldots,x_n$ and $\neg X=\neg x_1,\ldots,\neg x_n$. Let $\mcC=(W_1,W_2,\ldots,W_N)$ be the catalytic tape configuration where $N=\poly(n)$ will be suitably chosen. Here $W_i=(W_i(1),\ldots,W_i(n))$ is a weight assignment $W_i:[n]\to [n^2]$ on the literals. For a formula $\psi$, and a weight assignment $W:[n]\to [n^2]$, define the query$(\psi,W,w^*)$:
	\[\exists S: (x_j=1\iff j\in S)\wedge \psi(\bar x)=1\wedge W(S)=w^*.\] Again, the idea is to either find an isolating weight assignment or a threshold literal. The description of $\mcA$ is as follows:

        \begin{itemize}
        \item First round of $\NP$ queries: For all $i\in N$ do the following in: For all $w^*\in [n^3],j\in [n]$, ask the $\NP$ oracle queries of the form $q_1=$query$(\phi,W_i,w^*)$, $q_2=$query$(\phi\wedge(x_j=1),W_i,w^*)$, and $q_3=$query$(\phi\wedge (x_j=0),W_i,w^*-W_i(j))$. Let $b_1,b_2,b_3$ be the bits answered by the $\NP$ oracle. If $b_1$ is false for all $w^*$, we know that $\phi$ is not satisfiable, so return the empty formula as $\mcA(\phi)$ and halt. Otherwise, from the various answer bits $b_1$, corresponding to distinct values of weight $w^*$, we can infer $W_i^{min}$ to be the least value of $w^*$ for which $b_1=1$. Now, suppose for some $i$, $w^*=W_i^{min}$: $\forall j$ exactly one of $b_2$ and $b_3$ is true. In this case return $\mcA(\phi)=\phi\wedge \{W_i\}\wedge S=\{j:x_j=1\}\wedge W(S)=W_i^{min}$ and halt ($W_i$ is the isolating weight assignment). Otherwise, for each $i$ consider the least $j$ such that both $b_2$ and $b_3$ are true (i.e. $j$ is a threshold literal). Replace $W_i$ with $(j,W_i^j)$ where $W_i^j=(W_i(1),\ldots,W_i(j-1),W_i(j+1),\ldots,W_i(n))$ (note that this transformation can be performed in $\onew$-$\inplacefl$ since $j$ is written on the logspace worktape, and we need one left to right scan of the catalytic tape to shift the weights as necessary and write $j$ in the initial part of the block). Observe that in the process, we have saved $N\log n$ bits of space. Now we can run a $\PSPACE$ algorithm to go over all assignments of $\phi$. Let $S$ be such an assignment i.e. $x_j=1\iff j\in S$. Then return $\mcA(\phi)=\phi\wedge (x_j=1\iff j\in S)$ and move on to the next phase.
	 	\item Second round of  $\NP$ queries: At this stage for all $i\in [N]$ we have $(j,W_i^j)$, and we want to reconstruct $W_i$ from this using the $\NP$ oracle. Let $W_i^j(j)=0$. For all $i,w^*$, ask queries of the form $q_4=$query$(\phi\wedge(x_j=1),W_i^j,w^*)$, and $q_5=$query$(\phi\wedge (x_j=0),W_i^j,w^*)$ to which the answers are $b_4$ and $b_5$. Let $W_i^{min}$ be the least value of $w^*$ such that $b_5=1$, and $w_0$ be the least value of $w^*$ such that $b_4=1$. Then we have that $W_i(j)=W_i^{min}-w_0$. Hence, we have reconstructed $W_i$ completely.
	 \end{itemize}
	
\end{proof}

\begin{theorem}\label{thm:fsat}
	$\fF\in \cltonp$
\end{theorem}
\begin{proof}
	Let $\phi$ be the input formula over the literals $X=x_1,x_2,\ldots,x_n$ and $\neg X=\neg x_1,\ldots,\neg x_n$. Consider the following modified version of $\mcA$:
	\begin{itemize}
		\item First round of $\NP$ queries: Suppose for some $i$, $w^*=W_i^{min}$: $\forall j$ exactly one of $b_2$ and $b_3$ is true. In this case, return a satisfying assignment as $x_j=1\iff b_2=1$, and halt. Otherwise proceed as in $\mcA$, notice that in $\PSPACE$, we can return any satisfying assignment.
		\item Second round of $\NP$ queries: Identical as in $\mcA$.
	\end{itemize}
\end{proof}

We next show that $\BPP$, $\MA$ and $\ZPP^{\NP[1]}$  are contained in $\cltonp$,
building on the construction of pseudorandom generators based on hardness versus randomness \cite{NW,IW}. We recall the following well-known theorem of
Impagliazzo and Wigderson: If there is a language $L\in\E$ such that for almost all input lengths $n$, $L^{=n}$ requires\footnote{Nisan and Wigderson \cite{NW}
  showed this result under the seemingly stronger average-case hardness assumption for $L$.} boolean circuits of size $2^{\eps n}$ for some $\eps\in (0,1)$ then $\BPP=\P$. It was observed by Goldreich and Zuckerman \cite{GZ} that, as most $n$ bit strings interpreted as truth-tables $T:\{0,1\}^{\log n}\to \{0,1\}$ require circuits of $n^{\eps}$ for any constant $\eps\in (0,1)$ by Shannon's counting argument, such a truth-table $T$ can be randomly guessed
by a machine and its hardness verified by an $\NP$ oracle, and then the
hard truth-table $T$ can be used to derandomized $\BPP$ by \cite{NW,IW}.
That would yield the containments of $\BPP$ and $\MA$ in $\ZPP^\NP$.

\begin{theorem}[\cite{GZ}]\label{lemma:wig}
  $\MA$ is contained in $\ZPP^\NP$.
\end{theorem}

We next show that a $\CL$ base machine making queries to an $\NP$ oracle
suffices instead of a $\ZPP$ machine.

\begin{theorem}\label{thm:BppMaZppToNp}
	The following hold:
	\begin{enumerate}
		\item $\BPP\subseteq \cltonp$
		\item $\MA, \Co\MA\subseteq \cltonp$
		\item $\ZPP^{\NP[1]}\subseteq \cltonp$
	\end{enumerate}
\end{theorem}
\begin{proof}
  Throughout the proof, let $\eps$ be a fixed constant in $(0,1)$.
\begin{enumerate}
\item First, we show that $\BPP\subseteq \CL^{\fF}_{2\text{-round}}$. Let the catalytic tape configuration be $\mcC=(T_1,T_2,\ldots,T_N)$, $N=\poly(n)$ where each $T_i=\{T_i(x):x\in \{0,1\}^{\log n}\}$ is a truth table $T_i:\{0,1\}^{\log n}\to \{0,1\}$ that is a $n$ bit long string. We simulate $\BPP$ as follows:
		\begin{itemize}
			\item First Round of queries: For all $i\in [N]$, ask the query `Is there a circuit $C_i$ of size at most $n^{\eps}$ that computes $T_i$'. If for some $i$, the answer is `No', then we have found a hard truth table $T_i$. In that case, we can simulate the $\BPP$ algorithm using one $\NP$ query (by the argument sketched preceeding \Cref{lemma:wig}) and halt. Otherwise, we replace the catalytic tape with $(C_1,C_2,\ldots,C_N)$. Observe that we have freed up $(n-n^{\eps})\cdot N$ bits of space, so we can simulate the $\BPP$ algorithm in $\PSPACE$. Observe that the queries of the form $\exists C_i$ are $\fF$ queries, since we require access to such a $C_i$.
			\item Second Round of queries: At this stage, we have already solved the $\BPP$ algorithm, we want to revert the catalytic tape to its initial configuration. The current configuration is $(C_1,\ldots,C_N)$. For all $i\in[N],x\in\{0,1\}^{\log n}$, evaluate $C_i(x)$, and replace the catalytic configuration with $T_i=\{C_i(x):x\in\{0,1\}^{\log n}\}$. Each evaluation of $C_i(x)$ can be simulated by a $\NP$ query, since it is a $\P$ computation.
		\end{itemize}
		Now, we show how to replace the $\fF$ oracle with $\NP$ oracle using \Cref{thm:fsat}. Assume the catalytic tape is given as $\mcC=((T_1,W_1),(T_2,W_2),\ldots,(T_N,W_N))$ where $T_i$'s are truth tables as before, and $W_i$'s are designated to be weight assignments. When we ask `$\exists C_i$ circuit of size at most $n^{\eps}$' as an $\NP$ query, we do so by invoking Cook's Theorem, and $W_i$ is a weight assignment over the literals of that boolean formula representing this query. If there does exist a small enough circuit, we check whether or not $W_i$ isolates such a min-weight circuit: If Not, then we find a threshold literal and save $\log n$ bits of space in $W_i$ in the same way as in \Cref{thm:fsat}. If Yes, we store the min-weight circuit $C_i$ in place of $T_i$ and proceed as before. In the process, we will either find a hard truth table and derandomize $\BPP$, or save enough space to run the $\BPP$ algorithm in $\PSPACE$.
		\item Notice that if we can show $\MA\subseteq\cltonp$, the inclusion for $\Co\MA$ directly follows since $\cltonp$ is closed under complement. Thus, we proceed to show the inclusion for $\MA$. From the proof of $(1)$, we can assume that after the first round of $\NP$ queries made by the $\CL$ oracle, we have access to a truth table $T$ that is hard for all circuits of size at most $n^{\eps}$ (since otherwise, we can simulate $\PSPACE$). A language $L$ is in $\MA$ means $x\in L\iff \exists y : V(x,y)=1$ where $V$ is a $\BPP$ machine. From $T$ and \Cref{lemma:wig}, we can replace $V$ with a $\P$ machine $V_0$. Thus one more $\NP$ query will suffice to check if $x\in L$ (this $\NP$ query is $\exists y:V_0(x,y)=1$).
		\item Let $M$ be a $\ZPP^{\NP[1]}$ machine which behaves as follows: Given $x$, $M$ runs over its random bits $r$, after some polytime computation $M$ makes a $\SAT$ query $q$, and does some more polytime computation. Finally, with high probability, $M$ answers $x\in L$ (Yes) or $x\notin L$ (No) correctly, and with negligible probability outputs `Do Not Know'. The $\CL$ machine makes the following two $\NP$ queries in the first round: $q_1=\exists r,q: (M(x,r)\text{ queries $q$ to the }\NP \text{ oracle})\wedge q\in \SAT\wedge (M(x,r)\text{ returns Yes})$, $q_2=\exists r,q: (M(x,r)\text{ queries $q$ to the }\NP \text{ oracle})\wedge q\in \SAT\wedge (M(x,r)\text{ returns No})$. We call $x$ to be nice if either of the queries $q_1$ or $q_2$ is answered positively. Let $b_1,b_2$ be the answer bits (observe that both $b_1,b_2$ cannot be $1$ since a $\ZPP$ machine always answers correctly). If $b_1=1$, the $\CL$ machine returns $x\in L$, if $b_2=1$, it returns $x\notin L$. Otherwise $x$ is not nice, moreover we know that the $\NP$ oracle always answers negatively to the query $q$ made by $M$ on the random strings that return a Yes/No answer. Hence, we can simulate the $\ZPP$ algorithm as it is and set the answer to the $\NP$ query as False. Clearly, from $(1)$, $\ZPP\subseteq \cltonp$. Hence, we have $\ZPP^{\NP[1]}\subseteq \cltonp$ (here we make the queries $q_1,q_2$ along with the first round of queries required for the proof of $(1)$).
	\end{enumerate}
\end{proof}
\begin{figure}[h]
	\centering
	\includegraphics[width=10cm]{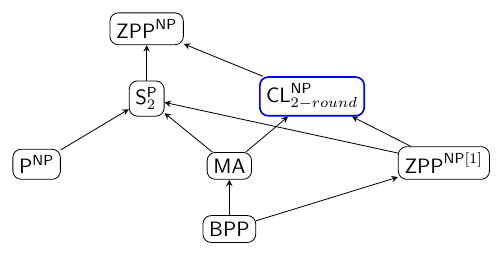}
	\caption{Inclusions of various oracle classes.}
	\label{fig:hasse}
\end{figure}
\Cref{fig:hasse} shows that $\S_2^\P$ and $\cltonp$ bear similar relationships to the classes $\MA,\ZPP^{\NP[1]}$ and $\ZPP^\NP$. The inclusions relating $\S_2^\P$ are shown in \cite{RS,Cai,CaiC}.

\section{Reachability in Catalytic Unambiguous Logspace}
\subsection{Preliminaries related to \cite{MP} for Reachability}\label{sec:blocks}

We first describe the technique developed in \cite{MP}.

\begin{definition}\label{def:isor}
	Let $G=(V,E)$ be a simple directed graph. A weight assignment $w:V\rightarrow \bbN$ is min-isolating for $G$ if for all $s,t\in V$, there is a unique min-weight path under the assignment $w$ from $s$ to $t$. Denote by $w(G,s,t)$ the weight of the min-weight $s-t$ path under $w$.
\end{definition}
We restrict attention to layered digraphs. Let $G=(V,E)$ be a layered directed graph, where $V=\sqcup_{0\le i\le d} V_i$ be the vertex partition into $d+1$ layers, where $|V|=n$. 
Edges in $G$ are only between adjacent layers, $V_i$ and $V_{i+1}$, $0\le i\le d$. Let $\ell=\ceil{\log d}$. For simplicity we assume $d=2^{\ell}$.

\begin{figure}[ht]
  \centering
  \includegraphics[width=14cm]{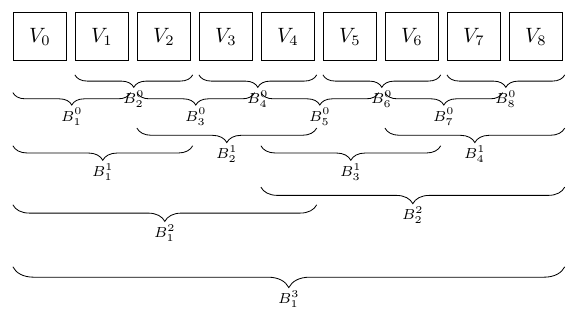}
  \caption{Block System of a layered DAG with $9$ layers.}
  \label{fig:blocks13}
\end{figure}

Consider the following \emph{Block System} of $G$: $B^0=(B^0_1,B^0_2,\ldots,B^0_{2^\ell})$ where $B^0_j=V_{j-1}\cup V_j$. Further, inductively define $B^i=(B^i_1,B^i_2,\ldots,B^i_{2^{\ell-i}})$ for $i\le \ell$ such that $B^i_j=B^{i-1}_{2j-1}\cup B^{i-1}_{2j}$ for all $j\le 2^{\ell-i}$, see \Cref{fig:blocks13}. In other words, $B^i$ can be seen as $2^{\ell-i}$ blocks of length $2^i+1$ each in $G$. For a block $B^i_j=\cup_{(j-1)\cdot 2^i\le k\le j\cdot 2^i}V_k$, the vertices in $V_{(j-1)\cdot 2^i}\cup V_{j\cdot 2^i}$ are the boundary vertices of the block, where $V_{(j-1)\cdot 2^i}$ is the left layer and $V_{j\cdot 2^i}$ is the right layer of the block, and all other vertices of $B^i_j$ are internal to it. \cite{MP} build weight assignemnts $w_i$ for each \emph{Block System} $B^i$ such that within each block $B^i_j$, $w_i$ is min-isolating for all $s,t$ paths where $s,t$ are in $B^i_j$. The middle layer of $B^i_j$ is $V_{(2j-1)\cdot 2^{i-1}}$. Clearly, setting $w_0\equiv 0$ is an isolating weight for $B^0$ because there is at most one edge between a pair of vertices between $V_{j-1}$ and $V_j$. In the weight assignment $w_i$, we always only assign nonzero weights to  the vertices internal to the blocks in $B^i$, i.e. to the vertices $V\setminus \cup_{j\le 2^{\ell-i}}V_{j\cdot 2^i}$, and boundary vertices $\cup_{j\le 2^{\ell-i}}V_{j\cdot 2^i}$ always have weight $0$. Now, we describe how \cite{MP} construct $w_{i+1}$ from $w_i$ inductively:

Let $B^{i+1}_j=B^i_{2j-1}\cup B^i_{2j}$ be a block in $B^{i+1}$. Assume that we are given a weight assignment $w_i$ that is zero in the left layer of $B^{i+1}_j$, the middle layer of $B^{i+1}_j$ i.e. $B^i_{2j-1}\cap B^i_{2j}$, and the right layer of $B^{i+1}_j$. Let the middle layer be $M=B^i_{2j-1}\cap B^i_{2j}$. Construct $w_{i+1}$ by extending $w_i$ by giving nonzero weights to these middle layers. We denote this subset of vertices as 
$$L_{i+1} = \cup_{\text {odd } j \in\left[2^{\ell-i}\right]} V_{j \cdot 2^i}.$$
Let $s,t\in B^{i+1}_j$ such that $s$ occurs in a layer earlier than $t$. The following  
cases arise:
\begin{enumerate}
	\item No $s$ to $t$ path crosses the middle layer of $B^{i+1}_j$. If $t$ is internal to $B^i_{2j-1}$ then $w_{i+1}$ is min-isolating for $s-t$ paths in $B^i_j$ no matter how $w_{i+1}$ assigns weights to $M=B^i_{2j-1}\cap B^i_{2j}$. This follows from the induction hypothesis that $w_{i+1}$ and $w_i$ agree on the vertices in $B^{i+1}_j\setminus M$. The case where $s$ is 
internal to $B^i_{2j}$ is similar.
	\item Otherwise, $s\in B^i_{2j-1}$ and $t\in B^i_{2j}$. Then, every path from $s$ to $t$ crosses the middle layer $M$. For vertices $x,y$ in a block of $B^i$, let $\mu_i(x,y)$ denote 
weight of the unique min-weight $x-y$ path under $w_i$. Let $\mu_i(x,y)=\infty$ if no such path exists. To guarantee that $w_{i+1}$ is min-isolating for $s-t$ paths in $B^{i+1}_j$, it suffices to have the following condition: for all vertices $u, v \in M$ that are on a path from $s$ to $t$,
	
	\begin{equation} \label{eq:3}
		\mu_i(s, u)+\mu_i(u, t)+w_{i+1}(u) \neq \mu_i(s, v)+\mu_i(v, t)+w_{i+1}(v)
	\end{equation}

 Condition \cref{eq:3} is the \emph{disambiguation requirement} \cite{MP}. Notice that this is a stronger condition than min-uniqueness.
\end{enumerate}

Now we explain the construction of weight assignment $w_{i+1}$ satisfying \cref{eq:3} using 
the following lemma.

\begin{lemma}[Universal hashing \cite{CW79}]\label{lemma:hash}
	There is a logspace computable weight assignment generator $\left(\Gamma_{m, r}\right)_{m, r \in \mathbb{N}}$ with seed length $s(m, r)=O(\log (m r))$ such that $\Gamma_{m, r}$ produces functions $h:[m] \mapsto[r]$ with the following property: For every $u, v \in[m]$ with $u \neq v$, and every $\delta \in \mathbb{N}$
	
	$$
	\operatorname{Pr}_h[h(u)=\delta+h(v)] \leq 1 / r
	$$
	
	where $h$ is chosen uniformly at random from $\Gamma_{m, r}$ i.e. $h=\Gamma_{m,r}(s)$ where $s$ is a uniformly at random chosen seed.
\end{lemma}
We identify $V$ with $[m]=[(d+1) \cdot n]$ in a natural way. If we pick $h: V \mapsto[r]$ uniformly at random from $\Gamma_{m, r}$ and set $w_{i+1}=h$ on $L_{i+1}$, then \Cref{lemma:hash} guarantees that each individual disambiguation requirement \cref{eq:3} holds with probability at least $1-1 / r$. As there are at most $n^4$ choices for $s, t, u, v$ (and $\delta=	-\mu_i(s, u)-\mu_i(u, t)+ \mu_i(s, v)+\mu_i(v, t)$ is a fixed constant only depending on $s,t,u,v$), all 
the disambiguation conditions are satisfied with probability at least $1-n^4 / r$ by a union bound. Therefore, choosing $r=\poly(n)$ suffices to get high success probability. Choosing $r=n^6$ we get 
a success probability $1-1/n^2$. 

By this technique, we can assign $w_i\equiv h_i$ on $L_i$ where $h_i$ are hash functions picked at random from $\Gamma_{n(d+1),r}$. Then with probability at least $1-\ell/n^2>1/n$, $w_{\ell}$ is a min-isolating weight assignment for $G$. Let $R_i$ denote the number of random bits required for the construction of $w_i$, and $W_i$ be the maximum weight of any vertex under the assignment $w_i$. Observe that the above construction yields $R_i=\mcO(\log n)+R_{i-1}, W_i=\mcO(\log n)$, i.e. $R_{\ell}=\mcO(\log ^2 n)$ and $W_{\ell}=\mcO(\log n)$. We require $\log n $ many hash functions for this construction of $w_{\ell}$. Now we explain how \cite{MP}  the number of hash functions used is reduced (i.e the number $R_{\ell}$) with an increase in the size of weights (i.e. $W_{\ell}$) using `shifting'.

\paragraph*{Hash and Shift:} Recall that the base case is $w_0\equiv 0$. Fix $\Delta=\ell^{\alpha}$ for some $\alpha\in[0,0.5]$ and $\gamma=2^{\ceil{\log (2n(d+1)r)}}$ (in fact \cite{MP} show $\gamma=\mcO(r)$ suffices). Given $w_i$, we pick a hash function $h$ from $\Gamma_{m,r}$, and use the same function for $w_{i+k}$ for all $k\le \Delta$ with a `shifting' procedure:  Set $w_{i+k}(v)=h(v)\cdot \gamma^{k-1}$ for all $v\in L_{i+k}$. These weights ensure that \cref{eq:3} is guaranteed for all $k\le \Delta$ with probability $\ge1-\Delta/n^2$ by \Cref{lemma:hash} (for details see \cite{MP}). Then, again repeat the same from $w_{i+\Delta+1}$ with a different hash function and continue. 

Hence, we have $R_{i+\Delta}=R_i+\mcO(\log(ndr))$ and $W_i\le r\cdot \gamma^{\Delta-1}\forall i$. Finally, we have $R_{\ell}=\mcO(\ell/\Delta\cdot \log n)$ and $W_{\ell}=\mcO(r\cdot n^{\mcO(\Delta)})$, i.e. $\log W_{\ell}=\mcO(\Delta\log n)$. Thus, $R_{\ell}=\mcO(\ell^{1-\alpha}\cdot \log n)$ and $\log W_{\ell}=\mcO(\ell^{\alpha}\log n)$.

We formalise the above construction in the following definition.

\begin{definition}\label{def:weight}
	For a fixed $\alpha\in[0,0.5]$, $\Delta=\ell^{\alpha}$ and $(h_1,\ldots,h_{k})\in \Gamma_{n(d+1),r}$ where $k\le \ell/\Delta$, define the weight assignment $w_{k\Delta}:V\rightarrow [r\cdot \gamma^{\Delta}]$ as follows:
	
    Set $w_0\equiv 0$. For an $i$, let $i-1=\Delta i_1+i_0$ be the unique representation where $i_1=\floor{(i-1)/\Delta},0\le i_0<\Delta$. Then $w_i(v)=h_{i_1+1}(v)\cdot \gamma^{i_0}$ for all $v\in L_i$, and $w_i$ agrees with $w_{i-1}$ everywhere else.
\end{definition}
The above discussion is summarized in the following Lemma.
\begin{lemma}[\cite{MP}]\label{lemma:goodhash}
	Given $(h_1,\ldots,h_k)\in \Gamma_{n(d+1),r}$ for some $k<\ell/\Delta$ such that $w_{k\Delta}$ is min-isolating for the blocks in the block system $B^{k\Delta}$ of $G$, we have that  
	
	$$\operatorname{Pr}_{h_{k+1	}}[w_{(k+1)\Delta}\text{ is not min-isolating for the block system }B^{(k+1)\Delta}]\le \Delta/n^2\le 1/n
	$$
	where $h_{k+1}$ is chosen uniformly at random from $\Gamma_{n(d+1),r}$.
\end{lemma}


\subsubsection*{Construction of weight assignments given the hash functions}

Here we give an algorithm to compute the intermediate weight functions using a given set of hash functions.
This is just rewriting \Cref{def:weight} as an algorithm, as follows:

Define $w_i\equiv 0$. Iterate $i'$ from $1$ to $i$, and on each iteration define $w_i(v)=h_{i_1'+1}(v)\cdot \gamma^{i_0'}$ for $v\in L_{i'}$ where $i'_1=\floor {(i'-1)/\Delta},i'_0=i'-1-\Delta i'_1$ as defined in \cref{def:weight}. Here $h_{i_1'+1}$ is the $i_1+1$th hash function. 
\begin{algorithm}
	\caption{ConstructWeights($i$)}\label{alg:weights0}
	\begin{algorithmic}[1]
		\State \textbf{Input:} $(G,i,\vec h=(h_1,\ldots,h_{i/\Delta}))$
		\State \textbf{Output:} Weight function $w_i$
		\State $w_i\equiv 0$
		\For{$1\le i'\le i$}
		\For {$v\in L_{i'}$}
		\State $i'_1\gets \floor {(i'-1)/\Delta},i'_0\gets i'-1-\Delta i'_1$
		\State $w_i(v)\gets h_{i_1'+1}(v)\cdot \gamma^{i_0'}$
		\EndFor
		\EndFor
		\State \Return $w_i$
	\end{algorithmic}
\end{algorithm}

The following is a lemma that we shall crucially use to $(1)$ detect whether a weight assignment is isolating or not, and $(2)$ given an isolating weight, compute distances.
\begin{lemma}[\cite{AR},\cite{MP}]\label{lemma:wcheck}
	There exist unambiguous nondeterministic machines $\mbox{WeightCheck}$ and $\mbox{WeightEval }$ such that for every digraph $G=(V, E)$ on $n$ vertices, weight assignment $w: V \mapsto \mathbb{N}$, and $s, t \in V$ :
	
	(i) $\mbox{WeightCheck}(G, w)$ decides whether or not $w$ is min-isolating for $G$, and
	
	(ii) $\mbox{WeightEval}(G, w, s, t)$ computes $w(G, s, t)$ (see \Cref{def:isor}) provided $w$ is min-isolating for $G$.
	
	Both machines run in time $\poly (\log (W), n)$ and space $\mcO(\log (W)+\log (n))$, where $W$ is the maximum weight of a vertex under $w$.
\end{lemma}
From \cref{lemma:wcheck}, we get an algorithm $\mbox{WeightCheck}(B^i,\vec h_{i/\Delta})$, that checks whether $w_i$ constructed from $\vec h_{i/\Delta}=(h_1,\ldots,h_{i/\Delta})$ is min-isolating for all blocks in the block system $B^i$ of $G$, that runs in time $\poly (\log (W_i), n)$ and space $\mcO(\log (W_i)+\log (n))$, where $W_i$ is the maximum weight of a vertex under $w_i$. This algorithm uses \Cref{alg:weights0} as a transducer.

\subsection{A Compress or Compute Algorithm}

Let $\Delta=\log^{\alpha}n$ for some $\alpha\in [0,0.5]$. We present the main 
\Cref{alg:comp1}, and its proof in \Cref{lemma:check}.

\begin{algorithm}[h]
	\caption{CompressOrCompute($G=(V,E)$)}\label{alg:comp1}
	\begin{algorithmic}[1]
		\State \textbf{Input:} $(G,\mcC=(\tau_1,\tau_2,\ldots,\tau_t))$ and an unambiguous oracle consisting of $\mcO(\log^{1+\alpha} n)$ space.
		\State \textbf{Output:} Either outputs a set of $\ell/\Delta$ hash functions, or directly solves the reachability instance using freed up space from the catalytic tape.
		\State $w_{0}\equiv 0$, $b_k=-1\forall k\in [t]$, $i=1,k=1$
		\For {$k\gets 1$ to $t$}\Comment{Process the $k$th element $\tau_k$ of the Catalytic Tape}
		\If {$i>\ell/\Delta$}
		\State break \Comment{We have found $\ell/\Delta$ good hash functions}
		\EndIf
		\State $b_k\gets 1$  \Comment{i.e. $h_i=\tau_k$}
		
		\If {$\mbox{WeightCheck}(B^{i\Delta},\vec h_i=\vec h_{i-1},h_i=\tau_k)$}
		\State $i\gets i+1$ \Comment{$h_i=\tau_k$ is a good hash function, next we want to find $h_{i+1}$}
		\Else 
		\State {$b_k\gets0$, $\tau_k\gets\comp(\mcC,\tau_k)$} \Comment{$h_i\gets\tau_k$ is a bad hash function}
		\EndIf
		\EndFor
		\If {$i>\ell/\Delta$}
		\State \Return $\vec h_{\ell/\Delta}$ \Comment{we have found the hash functions leading to a min-isolating weight assignment for $G$}
		\Else 
		\State Use the freed up space $\mcO(\log^{2-\alpha}n)$ space for the computation using the unambiguous oracle.
		\EndIf
		\For {$k\gets t$ downto $1$}
		\If {$b_k=0$}
		\State $\decomp(\mcC,\tau_k')$
		\EndIf
		\EndFor
	\end{algorithmic}
\end{algorithm}

\begin{theorem}\label{thm:nl}
	$\NL\subseteq \CTISP\left[\poly(n), \log n,\log^{2-\alpha}n\right]^{\U\TISP\left[\poly(n),\log^{1+\alpha}n\right]}$ for all $\alpha\in [0,0.5]$.
\end{theorem}
\begin{proof}
	As layered digraph reachability is $\NL$-complete (\Cref{nlhard}), it suffices to show a $ \CTISP\left[\poly(n), \log n,\log^{2-\alpha}n\right]^{\U\TISP\left[\poly(n),\log^{1+\alpha}n\right]}$ algorithm for directed reachability in layered digraphs. Let $(G,s,t)$ be an instance. 
	Given a min-isolating weight $W:V\to [\poly(2^{\log^{1+\alpha}n})]$ for $G$, using WeightEval$(G,s,t)$ \Cref{lemma:wcheck}, we can in $\UTISP\left[\poly(n), \log^{1+\alpha}n\right]$ check reachability from $s$ to $t$. Let $\ell=\mcO(\log n)$ and $\Delta=\ell^{\alpha}$. Now, suppose we have a $ \CTISP\left[\poly(n), \log n,\log^{2-\alpha}n\right]^{\U\TISP\left[\poly(n),\log^{1+\alpha}n\right]}$ algorithm to find $\ell/\Delta$ many hash functions $h_1,\ldots,h_{\ell/\Delta}:V\to [\poly(n)]$ such that the weight assignment $w_{\ell}$ (see \cref{def:weight}) is minimum isolating for $G$. Then we can send these hash functions $h_1,\ldots,h_{\ell/\Delta}$ to the $\U\TISP\left[\poly(n),\log^{1+\alpha}n\right]$ oracle. The oracle can now compute the weights $w_\ell$ as a $\log^{1+\alpha}n$ space transducer and use \emph{WeightEval}($G,w_\ell,s,t$) (\Cref{lemma:wcheck}) to compute reachability in $G$. This completes the proof, modulo \Cref{lemma:check}, which we set out to show next.
\end{proof}

At this point it is clear that we need to compute isolating weights of $\mcO(\log ^{1+\alpha}n)$ bits for a layered digraph $G =(V,E)$ with $n$ vertices in $  \CTISP\left[\poly(n), \log n,\log^{2-\alpha}n\right]^{\U\TISP\left[\poly(n),\log^{1+\alpha}n\right]}$ for all $\alpha\in [0,0.5]$. We do this in \Cref{alg:comp} closely following the technique in \cite{Pyne}. We first present an informal description of it. 

We are given a layered digraph $G=(V,E)$ with $d+1$ layers $V_0,\ldots,V_d$ where $d=2^{\ell}$, work space $\mcO(\log  n)$, catalytic tape is $\mcC=(\tau_1,\ldots,\tau_t)$ for $t=\mcO(\ell/\Delta)$ and each $\tau_k$ is $\mcO(\log n)$ bits and a $\U\TISP[\poly(n),\log^{1+\alpha}n]$ oracle. Let $B^i$ be the $i$th \emph{block system} of $G$ (recall definition from \Cref{sec:blocks}). For book-keeping, we have $b_1,\ldots,b_t\in\{-1,0,1\}$ in the work space of the base machine, initially all set to $-1$ ($b_j=1$ denotes that $\tau_j$ is a good hash function). We proceed in $\ell/\Delta$ steps, where in the $i$th step, we find a new hash function, and a weight assignment $w_{i\Delta}$ that is isolating for the block system $B^{i\Delta}$. We process $\tau_1,\ldots,\tau_t$ in the given order. We either sample hash functions from $\mcC$ or compress the catalytic tape. Suppose we are processing $\tau_k$, and we are searching for the $i^{th}$ hash function. Then $h_1,h_2,\ldots,h_{i-1}$ are as follows: Let $j_1,\ldots,j_{i-1}$ be the first $i-1$ indices such that $b_{j_1}=\ldots=b_{j_{i-1}}=1$. Then $h_1=\tau_{j_1},\ldots,h_{i-1}=\tau_{j_{i-1}}$. Now, we set $h_i=\tau_k$ and send the current set of hash functions to the oracle which can now compute $w_{i\Delta}$ of any vertex. Using $\mbox{WeightCheck}(B^{i\Delta},\vec h_i)$ \Cref{lemma:wcheck}, the oracle checks if the hash function found is good or not. If yes, then we set $b_k=1$ and proceed to find the next hash function and process $\tau_{k+1}$. Otherwise, we iterate over all hash functions that are bad to find the index $j$ such that $\tau_k$ is the $j$th bad hash function. We then replace $\tau_k$ with the index $j$, and save some space. In this case, set $b_k=0$ and proceed to find $h_i$, and process $\tau_{k+1}$. After processing all $\tau_k, 1\le k\le t$, either we 
have an isolating weight assignment $w_{\ell}$ for $G$, or we have saved enough space. If we have saved enough space, we find the necessary hash functions, and solve the problem. Then we restore the catalytic tape configuration by restoring the $j$th hash function wherever we had replaced the hash function with its index $j$. To do this replacement, we again invoke $\mbox{WeightCheck}$ \Cref{lemma:wcheck} since to iterate over all bad hash functions, we need to check if a weight is indeed isolating or not. A formal proof of the correctness of \Cref{alg:comp1} is given in the following lemma. 


\begin{lemma}\label{lemma:check}
	Given a layered digraph $G=(V,E)$ on $n$ vertices and $\alpha\in [0,0.5]$ where $\ell=\mcO(\log n)$ and $\Delta=\ell^{\alpha}$, we can in $ \CTISP\left[\poly(n), \log n,\log^{2-\alpha}n\right]^{\U\TISP\left[\poly(n),\log^{1+\alpha}n\right]}$ find $\ell/\Delta$ many hash functions $h_1,\ldots,h_{\ell/\Delta}:V\to [\poly(n)]$ such that the weight assignment $w_{\ell}$ (see \cref{def:weight}) is minimum isolating for $G$.

\end{lemma}
\begin{proof}
	Let $d=2^{\ell}$, where $d+1$ is the number of layers in $G$, and let $\Delta=\ell^{\alpha}=\mcO(\log^{\alpha}n)$. We will attempt to construct the weight function $W=w_{\ell}$ as in \Cref{def:weight}. For this we need access to $\ell/\Delta$ many `good' hash functions which we shall extract from the catalytic tape. For ease in notation, we assume that the each hash function of $\cref{lemma:hash}$, for $r=n^6$, can be represented using exactly $c\log n$ bits for a constant $c$ (there are $n^{c}$ strings of $c\log n$ bits). Assume that we are given a catalytic space $(c^2+c)(\ell/\Delta)\log n$, as $\mcC=(\tau_1,\tau_2,\ldots,\tau_t)$ where each $\tau_k$ is $c\log n$ bits and $t=(c+1)\ell/\Delta$. For each $k\in[t]$, we store a number $b_k\in\{-1,0,1\}$  in the worktape to indicate if $\tau_k$ is a `good' hash function or not (and $b_k=-1$ if $\tau_k$ has not been processed). We can do this since $t=\mcO(\log n)$. Now we justify the \Cref{alg:comp1}.
	\begin{enumerate}
		\item Base Case: Set $b_k=-1$ for all $k\in[t]$ and $\vec h=\phi$ designated to be the set of good hash functions. At any stage $\vec h$ is the sequence of $\tau_k$'s such that $b_k=1$ i.e. we do not explicitly store the set $\vec h$, but have enough information to access $\vec h$.
		\item Suppose we are processing the $k^{th}$ block of the catalytic tape, namely $\tau_k$ where $\mcC=(\tau_1',\tau_2'\ldots,\tau_{k-1}',\tau_k,\ldots,\tau_t)$ where the $\tau'$'s are either the same as $\tau$ or a compression of $\tau$. We have $\vec h_{i-1}=(h_1,h_2,\ldots,h_{i-1})$ where $h_1=\tau_{j_1},h_2=\tau_{j_2},\ldots,h_{i-1}=\tau_{j_{i-1}}$ where $j_1<j_2<\ldots<j_{i-1}$ are exactly the set of indices less that $k$ such that $b_{j_1}=\ldots=b_{j_{i-1}}=1$. We know that $\vec h_{i-1}$ is a sequence of `good' hash functions. Now we send $\vec h_{i-1}$ and $\tau_i$ to the oracle. The oracle does the followig: Set $h_{i}=\tau_k$ and define $w_{i\Delta}$ according to \cref{def:weight}, using \cref{alg:weights0}. Now, using $\mbox{WeightCheck}(B^{i\Delta},\vec h_i)$ (\cref{lemma:wcheck}), it checks if $w_{i\Delta}$ is min-isolating for the \emph{Block System} $B^{i\Delta}$ (can be done in $\mcO(\Delta\log n)$ space since $w_{i\Delta}$ is an $\mcO(i\Delta)$ bit weight function). The oracle reports whether this weight is indeed isolating or not, to the base machine. If this is the case, the base machine sets $b_{k}=1$ i.e. same as setting $\vec h_i=(h_1,h_2,\ldots,h_{i-1},h_{i}=\tau_k)$, set $\mcC=(\tau_1',\tau_2'\ldots,\tau_{k-1}',\tau_k'=\tau_k,\tau_{k+1},\ldots,\tau_t)$, and move to find the next hash function i.e. set $i=i+1$ and process $\tau_{k+1}$.
		
		 If this was not the case i.e. $\tau_k$ is a `bad' hash function, we compress using \Cref{alg:comp}: We know from \cref{lemma:goodhash} that the number of such `bad' hash functions is $n^c/n=n^{c-1}$ in number. Let $\bad(\vec h_{i-1})=(h'_1,h'_2,\ldots,h'_{n^{c-1}})$ be the sequence of such `bad' hash functions in increasing order as binary numbers. Now, for all $h\in \{0,1\}^{c\log n}$ check if $h\in \bad(\vec h_{i-1})$ by again sending $\vec h_{i-1}$ and $h$ to the oracle, which first constructs $w_{i\Delta}$ using $h$ as the new hash function $h_{i}=h$ in space $\Delta\log n$ (\cref{alg:weights0}), and then uses WeightCheck (\Cref{lemma:wcheck}). Thus, we find the index $j$ such that $\tau_k=h'_j$ i.e. $\tau_k$ is the $j$th `bad' hash function in $\bad(\vec h_i)$. We set $b_k=0$ and replace $\tau_k$ by $j$ i.e. $\mcC=(\tau_1',\tau_2'\ldots,\tau_{k-1}',\tau_k'=j,\tau_{k+1},\ldots,\tau_t)$. Observe that $\tau_k$ requires $c\log n $ bits to store, whereas $j$ is only $(c-1)\log n$ bits. Hence, we have saved $\log n$ bits in the process. Call this procedure\footnote{This skips making the free space contiguous; this can be achieved by shifting all the free space to the right, and while decompressing, re--shifting the space back} $\comp(\mcC,\tau_k)$ (\Cref{alg:comp}). After this compression, process $\tau_{k+1}$ with the aim of searching for the $i^{th}$ hash function.
		 
		  Suppose we have processed all the $\tau_k$'s, and solved reachability. Then we restore the catalytic tape. To do so, we enumerate over all $k\in[t]$ and check if $b_k=0$. Now, $\decomp$ can be described naturally. We want to decompress $\tau_k'$ where $b_k=0$, notice that $\tau_k'$ is a $(c-1)\log n$ bit number that is the index of the original bad hash function. Therefore $j=\tau_k'$ denotes the index such that $\tau_k$ (the original element in the Catalytic Tape) is the $j$th bad hash function. Now, as we can access the previous good hash functions on the catalytic tape (by checking if their corresponding $b=1$), we can enumerate over all $c\log n$ strings to find the $j$th hash function that is $\bad$, by checking if the corresponding weight assignment (which can be computed using \Cref{alg:weights0}) is isolating or not using \Cref{lemma:wcheck} (again by calling the oracle, like in $\comp$). Let $\tau_k$ be the $j$th such hash function. Then, define $\decomp(\mcC,j=\tau_k')=\tau_k$ (\Cref{alg:decomp}). 
		  
		  \item In the above process, suppose we found $\ell/\Delta$ `good' hash functions from the catalytic tape such that $W=w_{\ell}$ is indeed isolating, then we are done. 
		   Otherwise, we have compressed at least $t-\ell/\Delta$ many $\tau_k$'s. Therefore, we have saved at least $(t-\ell/\Delta)\log n$ bits of space, and from our choice of $t$ we have $((c+1)\ell/\Delta-\ell/\Delta)\log n=(\ell/\Delta).c\log n$ bits of free space available. Divide this free space into $\ell/\Delta$ many blocks of $c\log n$ bits. The $i^{th}$ block is designated to be the $i^{th}$ hash function $h_i$. Suppose we have found $h_1,\ldots,h_{i-1}$. Then we enumerate over all possible $n^c$ strings for $h_i$, and send $\vec h_{i-1},h_i$ to the oracle which then tells us if $h_i$ is good or not. If it is good, we move to the next block $i+1$, and otherwise we check if the next candidate for $h_i$ is good or not, and proceed in this way. After we have indeed found the required set of good hash functions, we send these hash functions to the oracle which then solves reachability for us (as in the proof of \Cref{thm:nl}). Then, we move on to reconstuct the catalytic tape as described in $\decomp$ above.
\end{enumerate} 
This completes the proof of correctness.
	
Observe that $\log W=\mcO(\ell^{\alpha}\log n)=\mcO(\log^{1+\alpha}n)$. In the above computations, the only places where we use the space beyond $\log n$ are storing the weight of a vertex, and when using \cref{lemma:wcheck}, both of which require only $\mcO(\log^{1+\alpha}n)$, and all these computations are only done in the oracle. The catalytic space used is $t\cdot c\log n=\mcO(\log^{2-\alpha}n)$ and workspace used is $\mcO(\log n)$. Thus, we are done.

\end{proof}
\begin{algorithm}
	\caption{Compress($\mcC,\tau_k$)}\label{alg:comp}
	\begin{algorithmic}[1]
		\State \textbf{Input:} $(G,b=(b_1,b_2,\ldots,b_t),\mcC=(\tau_1,\tau_2,\ldots,\tau_t))$
		\State \textbf{Output:} Compressed $\tau_k$
		\State $index\gets 0$, $hash\gets 0^{2\log n}$, $val\gets\tau_k$
		\State $i\gets 1+\#\{j|b_j=1,j<k\} $ \Comment{i.e. (the number of good hash functions before $\tau_k$) $+1$}
		\While {$hash\le val$}
		\State $\tau_k\gets hash$\Comment{in the catalytic tape set $\tau_k\gets hash$}
		\State $\vec h_i=\vec h_{i-1},h_i=\tau_k$
		\If {$\mbox{WeightCheck}(B^{i\Delta},\vec h_i)$ is False} \Comment{Using $h_i\gets hash$}
		\State $index\gets index+1$\Comment{$hash$ is the $index$th bad hash function}
		\EndIf
		\State $hash\gets hash+1$
		\EndWhile
		\State \Return $index$
	\end{algorithmic}
\end{algorithm}
\begin{algorithm}
	\caption{DeCompress($\mcC,\tau_k'$)}\label{alg:decomp}
	\begin{algorithmic}[1]
		\State \textbf{Input:} $(G,b=(b_1,b_2,\ldots,b_t),\mcC=(\tau_1',\tau_2',\ldots,\tau_t')),k$
		\State \textbf{Output:} Revert back the Catalytic Tape $\tau_k'$ to its initial value.
		\State $index\gets 0,$ $hash\gets 0^{2\log n}$
		\State $i\gets 1+\#\{j|b_j=1,j<k\} $ 
		\While {$index< \tau_k'$}
		\State $\tau_k\gets hash$\Comment{in the catalytic tape set $\tau_k\gets hash$}
		\State $\vec h_i=\vec h_{i-1},h_i=\tau_k$
		\If {$\mbox{WeightCheck}(B^{i\Delta},\vec h_i)$ is False}
		\State $index\gets index+1$\Comment{$hash$ is the $index$th bad hash function}
		\EndIf
		\State $hash\gets hash+1$
		\EndWhile
		\State Replace $\tau_k'$ with $\tau_k$
	\end{algorithmic}
\end{algorithm}

\begin{remark}\label{rem:nl}
	Notice that in the above algorithm, we never invoke the $\UTISP[\poly(n),\log^{1.5}n]$ algorithm of \cite{MP} as a blackbox. If we are able to find the necessary hash functions, we check for reachability using the unambiguous oracle. If, on the other hand, we are not able to find the necessary hash functions, we enumerate over all candidate hash functions to find the sequence of good hash functions. In this routine, we always check the goodness of a hash function using the unambiguous oracle. Since, with $\ell/\Delta$ many hash functions, the weights are of $\Delta\log n$ bitlength, we require $\log^{1+\alpha}n$ unambiguous space to check if a weight assignment is isolating via \Cref{lemma:wcheck}. Thus the oracle machine requires $\mcO(\log^{1+\alpha}n)$ bits of memory, and the catalytic space used is $\mcO(\ell/\Delta\cdot \log n)=\mcO(\log^{2-\alpha}n)$.
\end{remark}

\begin{corollary}\label{cor:nl}
	The following holds:
	\begin{enumerate}
		
		\item $\NL\subseteq \CTISP[\poly(n),\log n,\log^2n]\iff \UL\subseteq \CTISP[\poly(n),\log n,\log^{2}n]$.
		\item $\NL\subseteq \CSPACE[\log n,\log^2n]\iff \UL\subseteq \CSPACE[\log n,\log^{2}n]$.
		
	\end{enumerate}
	\begin{proof}
		\begin{enumerate}
			Setting $\alpha=0$ is \Cref{thm:nl} gives $\NL\subseteq \CTISP\left[\poly(n), \log n,\log^{2}n\right]^{\UL}$. It is not clear that if $\UL\subseteq \CTISP\left[\poly(n), \log n,\log^{2}n\right]$, then $\CTISP\left[\poly(n), \log n,\log^{2}n\right]^{\UL}\subseteq \CTISP\left[\poly(n), \log n,\log^{2}n\right] $. But the queries to the $\UL$ oracle that we make in the proof of \Cref{lemma:check} are always available on the catalytic tape, work tape and input tape (since the hash functions are available on the catalytic tape). So, if $\UL$ is contained in $\CTISP[\poly(n),\log n,\log^2 n]$ then we can simulate these oracle queries in the base machine itself. Thus if we have $\UL\subseteq \CTISP[\poly(n),\log n,\log^{2}n]$ then, $ \NL\subseteq \CTISP[\poly(n),\log n,\log^{2}n]$, and the converse follows since $\UL\subseteq\NL$. The same proof goes through for  $\NL\subseteq \CSPACE[\log n,\log^2n]\iff \UL\subseteq \CSPACE[\log n,\log^{2}n]$.

		\end{enumerate}
	\end{proof}
\end{corollary}


\section{Isolating Proof Trees}
In this section we revisit the isolation lemma as applied to the
evaluation of semi-unbounded boolean circuit. The problem is complete
for $\LogCFL$ which is known to be contained in catalytic logspace, as
$ \LogCFL \subseteq \TC^1 \subseteq \CL$. However, as in the previous
section we consider a variant of catalytic computation to explore
derandomizing the isolation lemma in this setting.

Recall that a boolean circuit $C$ is an acyclic digraph with gates as
vertices. Each gate is either an $\wedge$ gate or an $\vee$ gate. A
special gate is designated as the output gate and the in-degree $0$
gates are \emph{input gates}. Each input gate is labeled by a boolean
variable from the set
$\{x_1,x_2,\ldots,x_n, \bar{x}_1,\bar{x}_2,\ldots,\bar{x}_n\}$. In
general $\wedge$ and $\vee$ gates are allowed to have have two or more
inputs and the gates output can fan out as input to other gates in the
circuit.

\begin{definition}[\cite{MP}]
Let $C=(V, E)$ be a circuit, with $V$ denoting the set of gates and
the directed edge set is $E$. A weight assignment for $C$ is a mapping
$w: V \to \mathbb{N}$. The weight $w(F)$ of a \emph{proof tree} $F$
with output $v$ equals $w(v)$ plus the sum over all gates $u$ that
feed into $v$ in $F$, of the weight of the proof tree with output $u$
induced by $F$. For an input $z$ for $C$, and $g \in V, w(C, z, g)$
denotes the minimum of $w(F)$ over all proof trees $F$ for $(C, z,
g)$, or $\infty$ if no proof tree exists. The weight assignment $w$ is
min-isolating for $(C, z, g)$ if there is at most one proof tree $F$
for $(C, z, g)$ with $w(F)=w(C, z, g)$. For $U \subseteq V, w$ is
min-isolating for $(C, z, U)$ if $w$ is min-isolating for $(C, z, u)$
for each $u \in U$. We call $w$ min-isolating for $(C, z)$ if $w$ is
min-isolating for $(C, z, V)$.
\end{definition}

We only consider $\SAC^1$ circuits. These are boolean circuit of depth
$d=\mcO(\log n)$, where $\wedge$ gates have fanin $2$, and $\vee$ gates
have unbounded fanin. $V=\cup_{i\le d+1}V_i$, where $V_i\subset
[n]\times \{i\}$ and $V_0$ are the literals and constants. $\wedge$
and $\vee$ gates alternate depending on the parity of the
layer. Assume $d=2\ell$ for some $\ell$, let $V_{\leq i}
= \cup_{j \leq i} V_j$, and $L_i=V_{2i-1}$ denote the $i$th $\wedge$
layer.

\begin{lemma}[\cite{MP}]\label{lemma:treeweights}
	For fixed $n,\ell\in \bbN$ and $i\le i'\in \bbN$, let
	$D_{n,\ell}=[n]\times [2\ell+1]$ be naturally identified with
	$V$. If $w_i$ is a min-isolating weight assignment for
	$(C,z,V_{\le 2i})$, and $h$ is a hash function chosen
	uniformly at random from $\Gamma_{n(d+1),r}$
	(recall \Cref{lemma:hash}), then
		$$
	\operatorname{Pr}_h\left[w_{i'}\text { is min-isolating for }\left(C, z, V_{\leq 2 i^{\prime}}\right)\right] \geq 1-1 / n,
	$$
	where $w_{i+j}\forall j\le i-i'$ is defined with the help of the picked hash function $h$ as follows:
		$$
	w_{i+j}(g)= \begin{cases}w_{i+j-1}(g)+h(g) \cdot \gamma^{j-1} & \text { if } g \in L_{i+j} \\ w_{i+j-1}(g) & \text { otherwise }\end{cases}
	$$
	Here $\gamma, r$ are polynomials in $n$. The bitlength of $w_{i'}$ is $\mcO((i'-i)\log n+\text{bitlength of }w_i)$ and can be constructed in space $\mcO(\log n)$ given the access to the hash function $h$.
\end{lemma}
The analogous variant of \Cref{lemma:wcheck} in this context is the following:
\begin{lemma}[\cite{MP}, see also \cite{AR}]\label{lemma:wcheck0}
	There exist unambiguous nondeterministic machines
	$\mbox{WeightCheck}$ and $\mbox{WeightEval }, $each equipped
	with a stack that does not count towards the space bound, such
	that for every layered $\SAC^1$ circuit $C=(V, E)$ of depth
	$d$ with $n$ gates, every input $z$ of $C$, weight assignment
	$w: V \mapsto \mathbb{N}$, and $g \in V$ :
	
	(i) $\mbox{WeightCheck}(C,z,w)$ decides whether or not $w$ is min-isolating for $(C,z)$, and
	
	(ii) $\mbox{WeightEval}(C,z, w,g)$ computes $w(C,z,g)$ provided $w$ is min-isolating for $(C,z)$.
	
	Both machines run in time $\poly (2^d,\log (W), n)$ and space
	$\mcO(d+\log (W)+\log (n))$, where $W$ is the maximum weight
	of a gate under $w$.
\end{lemma}

		
As the main result of this section, we show a statement analogous for $\LogCFL$, mirroring \Cref{thm:nl} for $\NL$. For the right hand side of the inclusion, we show a tradeoff result, that $\LogCFL=\SAC^1$ can be simulated by a catalytic machine with workspace $s(n)=\log n$ and catalytic space $c(n)$ between $\log^{1.5}n$ and $\log^2 n$, with an oracled machine whose description is as follows: it is an unambiguous machine with work space $\log^3 n/c(n)$ and has an additional stack which is not counted for in the work space, while maintaining polynomial time. To prove this we require the following crucial lemma.

\begin{lemma}\label{lemma:sac1}
		Given a layered $\SAC^1$ circuit $C=(V,E)$ of depth
		$d$ with $n$ gates, with input $z$, and $\alpha\in
		[0,0.5]$, we can in
		$\C\TISP[\poly(n),\log{n},\log^{2-\alpha}{n}]^{\U\AuxPDA\text{-}\TISP[\poly(n),\log^{1+\alpha}n]}$ find a weight assignment $W:V\to
		[\poly(2^{\log^{1+\alpha}n})]$ that is min-isolating
		for $(C,z)$, and check if there exists a proof-tree
		$F$ of $(C,z)$.
\end{lemma}
\begin{proof}
	The proof is very similar to that of \Cref{lemma:check}. Fix
	$\Delta=\ell^\alpha=\mcO(\log^{\alpha}n)$. We shall
	use \Cref{lemma:treeweights} with $i'-i=\Delta$. Let
	$\mcC=(\tau_1,\tau_2,\ldots,\tau_t)$ where $t=(c+1)\ell/\Delta$,
	and each $\tau_k$ is $c\log n$ bits i.e. the catalytic space used is
	$(c^2+c)(\ell/\Delta)\log n$. Assume that the hash
	functions from \cref{lemma:hash} when $r=\poly(n)$ as
	in \Cref{lemma:treeweights}, can each be represented using
	exactly $c\log n$ bits (there are $n^c$ strings of $2\log n$
	bits). The $\tau_k$'s will serve as hash functions in our
	construction of a min-isolating weight. For each $k\in[t]$, we
	store a number $b_k\in\{-1,0,1\}$ in the worktape to indicate
	if $\tau_k$ is a `good' hash function or not (and $b_k=-1$ if
	$\tau_k$ is yet unprocessed).

\begin{enumerate}
\item
Base case: $w_0\equiv 0, b_k=-1\forall k\in[t]$. $\vec h=\phi$ is the
set of good hash functions. At any stage, $\vec h$ is the sequence of
$\tau_k$'s such that $b_k=1$.

\item Suppose, we are processing the $k$th block $\tau_k$, of the catalytic
tape. Let $\vec h_i=(h_1,\ldots,h_i)$. We send $\vec h_i$ and $\tau_k$ to the oracle. Note that we have the min-isolating
weight assignment $w_{i\Delta}$ for $(C,z,V_{\le 2i\Delta})$, and we
want to construct $w_{(i+1)\Delta}$. The oracled machine
applies \cref{lemma:treeweights} to construct $w_{(i+1)\Delta}$ with
access to the hash function $h_{i+1}=\tau_k$. Then it can check if
$w_{(i+1)\Delta}$ is min-isolating for $(C,z,V_{2(i+1)\Delta})$
using \cref{lemma:wcheck0}.

\begin{itemize}
\item \textbf{Case 1: $w_{(i+1)\Delta}$ is min-isolating.}~
In this case set $b_k=1$, which is equivalent to updating $\vec
h_{i+1}=(h_1,\ldots,h_i,h_{i+1}=\tau_k)$.  Set $k\leftarrow
k+1,i\leftarrow i+1$.

\item \textbf{Case 2: $w_{(i+1)\Delta}$ is not min-isolating.}~
Set $b_k=0$. In this case, $\tau_k$ is a bad hash
 function. By \cref{lemma:treeweights}, the number of such bad hash
 functions is at most $ n^c/n=n^{c-1}$. Let $\bad(\vec
 h_i)=(h'_1,h'_2,\ldots,h'_{n^{c-1}})$ be the sequence of the bad hash
 functions in increasing order, sorted as binary encoded
 integers. Now, for all $h\in \{0,1\}^{c\log n}$ check if
 $h\in \bad(\vec h_i)$ by querying the oracle which uses \cref{lemma:wcheck0} as we can
 construct $w_{(i+1)\Delta}$ using $h$ as the new hash function
 $h_{i+1}=h$ is space $\Delta\log n$. Thus, we can find the index $j$
 such that $\tau_k=h'_j$ i.e. $\tau_k$ is the $j$th bad hash function
 from $\bad(\vec h_i)$. Replacing the $k$th block of the catalytic
 tape with $j$ saves $\log n$ bits of space. This step is analogous to
 the $\comp$, \Cref{alg:comp}.
			
After completion of the computation, we invoke the $\decomp$
procedure: Suppose we want to decompress $\tau_k'$ where $b_k=0$. Then
$j=\tau_k'$ denotes the index such that $\tau_k$ (the original element
in the catalytic tape) is the $j$th bad hash function. As we can
access all the previous good hash functions on the catalytic tape (by
checking if their corresponding $b=1$), we can enumerate over all the
$c\log n$ strings to find the $j$th hash function that is $\bad$. Let
$\tau_k$ be the $j$th bad hash function. Then, we define
$\decomp(\mcC,j=\tau_k')=\tau_k$ (Checking whether a hash is good or bad is done by the oracle). As mentioned, this step is analogous
to \Cref{alg:decomp}.
\end{itemize}

\item
In the above process, if we found $\ell/\Delta$ good hash functions,
then we have successfully constructed $W=w_\ell$. Then,
by \cref{lemma:wcheck0}, if $\mbox{WeightEval}(C,z,W,g=\text{output
gate})<\infty$, we accept, and reject otherwise. If we do not find
$\ell/\Delta$ many good hash functions, we have compressed at least
$t-\ell/\Delta$ many $\tau_k$'s, saving at least $c(\ell/\Delta)\log
n$ space. In this space we can enumerate over all $c\log n$ strings to find the $i^{th}$ good hash function given $\vec h_{i-1}$ for all $i\le \ell/\Delta$ using the oracle to check if a hash function is good or not. Then using $\vec h_{\ell/\Delta}$, the oracle checks for a proof tree by invoking $\mbox{WeightEval}$, \Cref{lemma:wcheck0}. Finally, we reconstruct the catalytic tape as discussed above.

\end{enumerate}
This completes the algorithm description along with its proof of
correctness. Clearly, the oracle workspace is bounded by
$\mcO(\log^{1+\alpha}n)$ and catalytic space is bounded by
$\mcO(\log^{2-\alpha}n)$, and both the oracled and base machine runs in polynomial time.
\end{proof}
Again, notice that the above proof does not directly use \cite{MP}'s algorithm as a blackbox by the same reasoning as in \Cref{rem:nl}. Finally, we have the following result:
\begin{theorem}\label{thm:cir}
	$\LogCFL\subseteq \C\TISP[\poly(n),\log{n},\log^{2-\alpha}{n}]^{\U\AuxPDA\text{-}\TISP[\poly(n),\log^{1+\alpha}n]}$
	for all $\alpha\in[0,0.5]$.
\end{theorem}
\begin{proof}
	The proof follows from \Cref{sachard} and \Cref{lemma:sac1}.
\end{proof}
\begin{corollary}
For $\mcA\in\{ \CTISP[\poly(n),\log n,\log^{2}n],\CSPACE[\log n,\log^{2}n]\}$, we have $\LogCFL\subseteq \mcA\iff   \U\AuxPDA	\text{-}\TISP[\poly(n),\log n]\subseteq \mcA$.
\end{corollary}
\begin{proof}
	Setting $\alpha=0$ in \Cref{thm:cir}, we have 	$\LogCFL\subseteq \mcA^{\U\AuxPDA\text{-}\TISP[\poly(n),\log n]}$ for $\mcA\in\{ \CTISP[\poly(n),\log n,\log^{2}n],\CSPACE[\log n,\log^{2}n]\}$. Observe that if $\U\AuxPDA\text{-}\TISP[\poly(n),\log n]$ is contained in $\mcA$, then we can perform the oracle queries made in \Cref{lemma:sac1} in the base machine itself, since the queries are already present in the base machine's catalytic tape and input tape. Thus $\LogCFL\subseteq \mcA\iff   \U\AuxPDA	\text{-}\TISP[\poly(n),\log n]\subseteq \mcA$ follows. 
\end{proof}

\section{Concluding remarks}
We have shown new applications of the \emph{compress-or-random} pradigm in
catalytic computation. In particular, using the isolation lemma we obtain 
a general search to decision reduction that is computable in catalytic 
logspace, thereby showing several natural search problems are in 
catalytic logspace. With a different use of this paradigm, we obtain 
containments of several complexity classes in $\cltonp$, showing that
$\cltonp$ is closely related to $\ZPP^\NP$. Finally, we explore the
question of simulating $\NL$ using $\mcO(\log^2 n)$ catalytic space and
similar catalytic space upper bounds for $\LogCFL$.

We are left with many intriguing questions, such as:
\begin{enumerate}
\item Can Savitch be made purely catalytic i.e. can we prove:
$\NL\subseteq \CSPACE[\log n,\log^2 n]$? Notice that from \Cref{cor:nl}, it suffices to show that $\UL\subseteq \CSPACE[\log n,\log^2 n]$.
\item What is the relationship between $\cltonp$ and $\stwop$? They seem to have
similar lower and upper bounds.
	\item  For the containments in $\cltonp$ of the classes $\BPP$, $\MA$, 
and $\ZPP^{\NP[1]}$ shown in Section~\ref{sec:oracle} can we improve it to containment
in $\CL^{\NP}_{1\text{-round}}$?
\end{enumerate}

\section*{Acknowledgements}
We would like to thank Nathan Sheffield for interesting questions and
valuable comments on the previous version of this work that helped us
clarify some definitions and rectify some errors.

\bibliographystyle{alpha}
\bibliography{biblio}

\end{document}